\documentclass[preprint,12pt,envcountsame]{elsarticle} %
\biboptions{sort&compress}

\def\showappendix{}

\usepackage{ifthen}
\usepackage{xargs}

\usepackage{color}
\usepackage{amsmath,amsfonts,amsthm}
\usepackage{txfonts}

\usepackage{tikz}
\tikzstyle{max}=[shape=rectangle,draw,inner sep=0pt,minimum size=6mm,thick]
\tikzstyle{ran}=[shape=circle,draw,inner sep=0pt,minimum size=6mm,thick]

\newcommand{\Zset}{\mathbb{Z}}

\newcommand{\A}{\mathcal{A}}
\newcommand{\G}{\mathcal{G}}
\newcommand{\Qset}{\mathbb{Q}}
\newcommand{\eps}{\varepsilon}
\newcommand{\size}[1]{||#1||}

\newcommand{\genTran}[2]{%
    {}\mathchoice%
    {\stackrel{#1}{#2}}
    {\mathop {\smash{#2}}\limits^{\vrule width 0pt height 0pt depth 4pt\smash{#1}}}
    {\stackrel{#1}{#2}}
    {\stackrel{#1}{#2}}
{}}
\newcommand{\tran}[1]{\genTran{#1}{\rightarrow}}
\newcommand{\btran}[1]{\genTran{#1}{\leadsto}}
\newcommand{\sP}[1]{#1_P}  %
\newcommand{\sMax}[1]{#1_\top}  %
\newcommand{\sMin}[1]{#1_\bot}  %
\newcommand{\SSG}[3]{(#1,#2,#3)}  %
\newcommand{\ssgG}{\SSG{S}{\btran{}}{Prob}}
\newcommand{\OCSSG}[3]{(#1,#2,#3)}  %
\newcommand{\OCMDP}{\OCSSG}
\newcommand{\rules}{\Delta}
\newcommand{\Prob}{\mathit{Prob}}
\newcommand{\ocssgA}{\OCSSG{Q}{\rules}{P}}
\newcommand{\ocmdpA}{\ocssgA}
\newcommand{\conf}[2]{(#1,#2)}  %
\newcommand{\confs}[1]{{#1}{\times}\Zset}
\newcommand{\stateIII}[3]{\langle#1,#2,#3\rangle}  %
\newcommand{\stateV}[5]{[#1,#2,#3,#4,#5]}  %
\newcommand{\stratCL}{counterless}
\newcommand{\len}[1]{len(#1)}
\newcommand{\run}{Run}  %
\newcommand{\runs}[1]{\run{}(#1)}  %
\newcommand{\CState}[1]{\text{State}^{(#1)}} %
\newcommand{\Cnt}[1]{\text{C}^{(#1)}} %
\newcommand{\Prb}[2]{\mathbb{P}^{#1}_{#2}}
\renewcommand{\Pr}[3]{\Prb{#1}{#2}\hspace{-0.16em}\left({#3}\right)}   %
\newcommand{\Exp}{\mathbb{E}}
\newcommand{\Ex}[3]{\Exp^{#1}_{#2}\hspace{-0.16em}\left[{#3}\right]}   %
\newcommand{\Term}{\mathit{Term}}
\newcommand{\Reach}[1]{\mathit{Reach}_{#1}}
\newcommand{\ReachT}{\Reach{T}}
\newcommandx{\valO}[3][1=]{\mathrm{Val}_{#1}({#2},{#3})}  %
\newcommandx{\vt}[3][1=]{\mathrm{Val}_{#1}(\Term,\conf{#2}{#3})}  
\newcommand{\LimInf}[2]{{\mathit{LimInf}(#1 #2)}}
\newcommand{\CN}{{\LimInf{=}{{-}\infty}}}
\newcommand{\opi}{\pi^{*}} %
\newcommand{\SSGReach}{\text{reach $s_0$}}
\newcommand{\dtr}[3]{\mathit{DTR}_{#2}^{#1}(#3)}  %
\newcommand{\mindtr}[2]{\dtr{#1}{#2}{*}}  %
\newcommand{\TermB}[1]{\Term^{\leq{#1}}}  %
\newcommand{\vtb}[3]{\valO{\TermB{#3}}{\conf{#1}{#2}}}  %

\newcommand{\NPcoNP}{\mathrm{NP}\cap\mathrm{coNP}}

\newcommand{\sL}{\mathcal{L}}
\newcommand{\mar}[1]{m^{(#1)}}
\newcommand{\poten}{z} %
\newcommand{\vdiff}{\bar{\poten}_{\max}}
\newcommand{\qd}{\mathsf{trap}}
\newcommand{\ms}[1]{m^{(#1)}}

\newcommand{\expo}[1]{\mathrm{e}^{#1}}

\newcommand{\ifApp}[2]%
{\ifthenelse{\isundefined{\showappendix}}{#2}{#1}}

\newcommand{\myqed}{}

\theoremstyle{plain}
\newtheorem{theorem}{Theorem}[section]
\newtheorem{corollary}[theorem]{Corollary}
\newtheorem{lemma}[theorem]{Lemma}

\newtheorem{fact}[theorem]{Fact}
\newtheorem{claim}{Claim}

\theoremstyle{definition}
\newtheorem{example}[theorem]{Example}
\newtheorem{definition}[theorem]{Definition}

\theoremstyle{remark}
\newtheorem{remark}[theorem]{Remark}

\newcounter{repthmcnt}

\title{Approximating the Termination Value of One-Counter MDPs and Stochastic Games}

\author[fimu,fn1]{Tom\'{a}\v{s} Br\'{a}zdil}
\ead{brazdil@fi.muni.cz}
\author[uoe,fn1]{V\'{a}clav Bro\v{z}ek}
\ead{vbrozek@inf.ed.ac.uk}
\author[uoe]{Kousha Etessami}
\ead{kousha@inf.ed.ac.uk}
\author[fimu,fn1]{Anton\'{\i}n Ku\v{c}era}
\ead{kucera@fi.muni.cz}
\address[fimu,fn1]{Faculty of Informatics, Masaryk University,\\
Botanick\'a 68a, 60200 Brno\\
Czech Republic
}
\address[uoe]{School of Informatics, University of Edinburgh \\
Informatics Forum\\
10 Crichton Street\\
EH8 9AB, Edinburgh\\
United Kingdom}
\fntext[fn1]{Authors supported by the research center 
Institute for Theoretical 
Computer Science, project No.~1M0545. Tom\'{a}\v{s} Br\'{a}zdil and 
Anton\'{\i}n Ku\v{c}era are also supported by the Czech Science Foundation, 
project No.~P202/10/1469. V\'{a}clav Bro\v{z}ek is also supported by
Newton International Fellowship from the Royal Society.}

\begin{document}

\begin{abstract}
One-counter MDPs (OC-MDPs) and one-counter simple stochastic games (OC-SSGs)
are 1-player, and 2-player turn-based zero-sum, stochastic games played on
the transition graph of classic one-counter automata
(equivalently, pushdown automata with a 1-letter stack alphabet).
A key objective for the analysis and verification of these 
games is the \emph{termination}
objective, where  the players aim  to maximize (minimize, respectively) 
the probability of hitting counter value~$0$, 
starting at a given control state and given counter value.

Recently,
we studied \emph{qualitative} decision problems (``is the optimal
termination value equal to~$1$?'')  for OC-MDPs (and OC-SSGs) and
showed them to be decidable in polynomial time (in $\NPcoNP$, respectively).
However, \emph{quantitative} decision and approximation problems (``is
the optimal termination value at least~$p$'', or ``approximate the
termination value within~$\eps$'') are far more challenging.  This is
so in part because optimal strategies may not exist, and because even
when they do exist they can have a highly non-trivial structure.  It
thus remained open even whether any of these quantitative termination
problems are computable.

In this paper we show that all quantitative \emph{approximation}
problems for the termination value for OC-MDPs and OC-SSGs are
computable.  Specifically, given a OC-SSG, and given $\eps > 0$, we
can compute a value~$v$ that approximates the value of the OC-SSG
termination game within additive error~$\eps$, and furthermore we can
compute $\eps$-optimal strategies for both players in the game.

A key ingredient in our proofs is a subtle martingale, derived from
solving certain linear programs that we can associate with a maximizing OC-MDP.
An application of Azuma's inequality on these martingales yields a
computable bound for the ``wealth'' at which a ``rich person's
strategy'' becomes $\eps$-optimal for OC-MDPs.

\end{abstract}

\maketitle

\section{Introduction}
\label{sec-intro}
In recent years, there has been substantial research done to
understand the computational complexity of analysis and verification
problems for classes of finitely-presented but infinite-state
stochastic models, MDPs, and stochastic games, whose transition graphs
arise from basic infinite-state automata-theoretic models, including:
context-free processes, one-counter processes, and pushdown processes.
It turns out these models are intimately related to important
stochastic processes studied extensively in applied probability
theory.  In particular, one-counter probabilistic automata are
basically equivalent to (discrete-time) quasi-birth-death processes
(QBDs) (see \cite{EWY08}), which are heavily studied in queuing theory
and performance evaluation as a basic model of an unbounded queue with
multiple states (phases).  It is very natural to extend these purely
probabilistic models to MDPs and games, to model adversarial queuing
scenarios.

In this paper we continue this work by studying quantitative
\emph{approximation} problems for \emph{one-counter MDPs (OC-MDPs)}
and \emph{one-counter simple stochastic games (OC-SSGs)}, which are
1-player, and turn-based zero-sum 2-player, stochastic games on
transition graphs of classic one-counter automata.  In more detail, an
OC-SSG has a finite set of control states, which are partitioned into
three types: a set of \emph{random} states, from where the next
transition is chosen according to a given probability distribution,
and states belonging to one of two players: \emph{Max} or \emph{Min},
from where the respective player chooses the next transition.
Transitions can change the state and can also change the value of the
(unbounded) counter by at most $1$.  If there are no control states
belonging to \emph{Max} (\emph{Min}, respectively), then we call the
resulting 1-player OC-SSG a \emph{minimizing} (\emph{maximizing},
respectively) OC-MDP.  Fixing strategies for the two players yields a
countable state Markov chain and thus a probability space of infinite
runs (trajectories).

A central objective for the analysis and verification of OC-SSGs, is
the \emph{termination} objective: starting at a given control state
and a given counter value $j > 0$, player Max (Min) wishes to maximize
(minimize) the probability of eventually hitting the counter value $0$
(in any control state).  From well know fact, it follows that these
games are \emph{determined}, meaning they have a \emph{value}, $\nu$,
such that for every $\eps>0$, player Max (Min) has a strategy that
ensures the objective is satisfied with probability at least
$\nu-\eps$ (at most $\nu+\eps$, respectively), regardless of what the
other player does.  This value can be \emph{irrational} even when the
input data contains only rational probabilities, and this is so even
in the purely stochastic case of QBDs without players (\cite{EWY08}).

A special subclass of OC-MDPs, called
\emph{solvency games},
was 
studied in \cite{BKSV08} as a simple model of risk-averse 
investment.  
Solvency games correspond to OC-MDPs
where there is only one control state, but 
there are multiple actions that change the counter
value (``wealth''), possibly by more than 1 per transition, 
according to a finite support probability distribution 
on the integers associated with each
action.  The goal is to minimize the probability of going bankrupt,
starting with a given positive wealth.
It is not hard to see that these are subsumed by minimizing OC-MDPs
(see \cite{BBEKW10}).
It was shown in \cite{BKSV08} that if the
solvency game satisfies a number of restrictive assumptions
(in particular, on the eigenvalues of a matrix
 associated with the game),
then an optimal ``rich person's'' strategy (which does the same action
whenever the wealth is large enough)  can be computed for it 
(in exponential time).
They showed such strategies are not optimal
for unrestricted solvency games and left the unrestricted case 
unresolved in \cite{BKSV08}.

We can classify analysis problems for OC-MDPs and OC-SSGs into two kinds.
\emph{Quantitative} analyses, which include: 
``is the game value at least/at most $p$'' for a given $p \in [0,1]$; or
``approximate the game value'' to within a desired additive
error $\eps > 0$.
We can also restrict ourselves to 
\emph{qualitative} analyses, which asks
``is the game value = 1? = 0?''.\footnote{The problem 
``is the termination value = 0?'' 
is easier, and can be solved in polynomial time without
even looking at the probabilities labeling the transitions of the OC-SSG.}
We are also interested in strategies
(e.g., memoryless, etc.)
that achieve these.

In recent work \cite{BBEKW10,BBE10}, we have studied \emph{qualitative}
termination problems
for
OC-SSGs.  
For both \emph{maximizing} and \emph{minimizing} OC-MDPs,
we showed that these problems are decidable in P-time,
using linear programming, connections to the theory of 
random walks on integers, and other MDP objectives.
For OC-SSGs, we showed the qualitative termination problem 
``is the termination value = 1?'' is 
in NP $\cap$ coNP.  This problem is already as
hard as Condon's quantitative termination problem for finite-state SSGs.   
However we left open, as the main open question,  
the computability of \emph{quantitative} termination
problems for OC-MDPs and OC-SSGs.

\textbf{Our contribution.}
In this paper, we resolve positively the computability of all quantitative
\emph{approximation} problems associated with OC-MDPs and OC-SSGs.
Note that, in some sense, approximation 
of the termination value in the setting of OC-MDPs and OC-SSGs
can not be avoided.  This is so not only because the value can
be irrational, but
because (\ifApp{see Example~\ref{ex:noopt} in 
Section~\ref{sec:non-opt-ex}}{see~\cite{fullversion}})
for maximizing
OC-MDPs there need not exist any optimal strategy for maximizing
the termination
probability, only $\eps$-optimal ones (whereas Min does have an optimal
strategy in OC-SSGs).
Moreover, even for minimizing OC-MDPs, 
where optimal strategies do exist, they
can have a very complicated structure.  In particular, 
as already mentioned for solvency games, 
there need not exist any ``rich person's'' strategy
that can ignore the counter value when it is larger than some finite $N \geq 0$.

Nevertheless, we show all these difficulties
can be overcome
when the goal is to \emph{approximate}
the termination value of OC-SSGs and to compute $\eps$-optimal 
strategies. Our \emph{main result} (Theorem \ref{thm:main}) is the following:

\begin{trivlist}\it
\item There is an algorithm that,
given as input: a OC-SSG, $\G$,
an initial control state $s$,  an initial counter
value $j > 0$,  and a (rational) approximation threshold $\eps > 0$,
\begin{itemize}  
\item computes a rational number, $v'$, 
  such that $|v' - v^* | < \eps$,  where $v^*$ is  the value of the  
  OC-SSG termination game on~$\G$, starting in configuration $(s,j)$, and
\item computes $\eps$-optimal 
  strategies  for both players in the OC-SSG termination game. 
\end{itemize}
For OC-MDPs, i.e., 1-player OC-SSGs,
the algorithm runs in exponential time in the encoding size of 
the OC-MDP,
and in polynomial time in $\log(1/\eps)$ and $\log(j)$.
For 2-player OC-SSGs, the algorithm runs
in nondeterministic exponential time in the encoding size of the OC-SSG.\footnote{We shall explain after the statement of Theorem \ref{thm:main},
in footnote \ref{foot:nondet-explain}, p 
precisely what we mean by computing something 
in {\em nondeterministic}
exponential time.  It amounts to the standard notion of nondeterministic
computation used in the setting of total search problems.}
\end{trivlist}

We now outline our basic strategy for proving this theorem.
Consider the case of maximizing OC-MDPs, and suppose we
would like to approximate the optimal termination probability,
starting at state $q$ and counter value $i$.
Intuitively, it is not hard to believe that as the counter value %
goes to infinity,  
the optimal probability of termination starting at a state $q$
begins to approach the optimal probability, $v_q$, of forcing
the counter to have a $\liminf$ value $= -\infty$.
We prove that this is indeed the case.
But we can compute the optimal value $v_q$ and an optimal strategy
for achieving it, based on results in our prior work \cite{BBEKW10,BBE10}.
For a given $\eps > 0$, 
we need to compute a bound~$N$ on
the counter value, such that
for any state~$q$, 
and all counter values $N' > N$, the optimal
termination probability starting at $(q,N')$ is at most
$\eps$ away from the optimal probability for the counter
to have $\liminf$ value $= -\infty$.
\emph{A priori} it is not clear whether such a bound
$N$ is computable, although it is clear that $N$ exists.  
To show that it is computable, we employ a subtle
(sub)martingale, derived from solving a certain linear programming problem
associated with a given OC-MDP.    By applying Azuma's
inequality on this martingale, we are able to show there are
computable values $c < 1$, and $h \geq 0$, such that 
for all $i > h$,  
starting from a state $q$ and counter value $i$, 
the optimal probability of both terminating
and not encountering any state from which with probability 1 the player
can force the $\liminf$ counter value to go to $-\infty$,
is at most $c^i/(1-c)$.  Thus, the 
optimal termination probability
approaches from above
the 
optimal probability of forcing the $\liminf$ counter value
to be $- \infty$, and 
the difference between these two values is exponentially small in $i$, with a computable base $c$. 
This martingale argument extends to OC-MDPs an argument recently 
used in \cite{BKK11} for analyzing purely probabilistic one-counter automata
(i.e., QBDs).

These bounds allow us to reduce the problem of approximating
the termination value to the reachability problem for an
exponentially larger finite-state MDP, which we can solve
(in exponential time) using linear programming.
The case for general OC-SSGs and minimizing OC-MDPs turns out to 
follow a similar line of argument, reducing the essential problem to
the case of maximizing OC-MDPs.
In terms of complexity, the OC-SSG case 
requires ``guessing''
an appropriate (albeit, exponential-sized) strategy,
whereas the relevant exponential-sized strategy can
be computed in deterministic exponential time for OC-MDPs.
So our approximation algorithms run in exponential time for OC-MDPs
and nondeterministic exponential time for OC-SSGs.

\textbf{Related work.}
As noted, 
one-counter automata with a non-negative counter are equivalent 
to  pushdown automata restricted to a 1-letter stack alphabet
(see \cite{EWY08}), and  
thus OC-SSGs with the termination objective form a subclass of
pushdown stochastic games, or equivalently, Recursive 
simple stochastic games (RSSGs).
These more general stochastic games were studied in~\cite{EY05icalp}, 
where it was shown that
many interesting computational 
problems, including any nontrivial approximation
of the termination value  for general RSSGs and RMDPs is undecidable,
as are qualitative termination problems.
It was also shown in
\cite{EY05icalp} that for  stochastic context-free games (1-exit RSSGs),
which correspond to pushdown stochastic games with only one state,
both qualitative and quantitative termination problems are decidable,
and in fact qualitative termination problems are decidable in NP$\cap$coNP
(\cite{EY06stacs}), while quantitative termination problems are 
decidable in PSPACE.
Solving termination objectives is a key ingredient for many more
general analyses and model checking problems for such stochastic games
(see, e.g., \cite{BBFK06,BBKO09}).
OC-SSGs are incompatible with stochastic context-free games.
Specifically, for OC-SSGs,
the number of stack symbols is bounded by 1, instead of the number of control states.

MDP variants of QBDs, essentially 
equivalent to OC-MDPs, have been considered in 
the queueing
theory 
and stochastic modeling literature, see \cite{White05,LHB07}.
However, in order to keep their analyses tractable, these works 
perform 
a naive finite-state ``approximation'' by 
cutting off the value of the counter at
an arbitrary finite value $N$, and adding \emph{dead-end absorbing} states
for counter values higher than $N$.
Doing this can radically alter the
behavior of the model, even for purely probabilistic QBDs,
and these authors establish no rigorous
approximation bounds for their models.
In a sense, our work can be seen as a much more careful and
rigorous approach to finite approximation, employing
at the boundary other objectives like maximizing the probability that 
the $\liminf$ counter
value $= -\infty$.  Unlike the prior work we establish
rigorous bounds on how well our finite-state 
model approximates the original infinite OC-MDP.

\vspace*{-0.1in}

\section{Preliminaries}

We assume familiarity with basic notions from probability theory.
We call a probability distribution
$f$ over a discrete set, $A$, \emph{positive} if $f(a) > 0$ for all
$a \in A$.

\begin{definition}
A \emph{One-Counter Simple Stochastic Game (OC-SSG)}
is given as
$\A=\ocssgA$, where
\begin{itemize}
\item $Q$ is a finite non-empty set of \emph{control states}, partitioned
into the states $\sMax{Q}$ of player Max,
$\sMin{Q}$ of player Min, and stochastic states $\sP{Q}$;
\item a set $\rules\subseteq Q\times\{-1,0,+1\}\times Q$
of \emph{transition rules},
such that for all $q\in Q$ there is some $(q,a,r)\in\delta$;
\item a map $P$ taking each tuple $(q,a,r)\in \rules$ with $q\in \sP{Q}$ to
a positive rational number $P((q,a,r))$, so that
for every $q\in \sP{Q}$:
$\sum_{(q,a,r)\in\delta}P((q,a,r))=1$.
\end{itemize}
A \emph{configuration} is a pair $\conf{q}{c}$ of a control state, $q$,
and an integer counter value $c\in\Zset$.
The set of all configurations is $\confs{Q}$.
An OC-SSG where $\sMin{Q}=\emptyset$ is called a \emph{maximizing 
 One-Counter Markov Decision Process (maximizing OC-MDP)},
similarly $\sMax{Q}=\emptyset$ defines a \emph{minimizing OC-MDP}.
Finally, if $\sMax{Q}=\sMin{Q}=\emptyset$ we have a
\emph{One-Counter Markov Chain (OC-MC)}.
\end{definition}

Let us fix a OC-SSG, $\A=\ocssgA$.
A \emph{run} in $\A$ is an infinite sequence
of configurations
\(
\omega
=
\conf{q_0}{c_0}
\conf{q_1}{c_1}
\cdots
\)
such that for all $i\geq 1$ we have that $(q_{i-1},c_i-c_{i-1},q_i)\in \rules$.
We define for every $n\geq 0$ the following functions:
\begin{itemize}
\item
$\CState{n}: \run \to Q$ returns the $n$-th control state:
$\CState{n}(\omega)=q_n$.
\item
$\Cnt{n} : \run \to \Zset$
returns the $n$-th counter value:
$\Cnt{n}(\omega)=c_n$.
\end{itemize}

A finite prefix,
\(
w
=
\conf{q_0}{c_0}
\cdots
\conf{q_k}{c_k}
,
\)
of a run
is called a \emph{finite path},
and $\len{w}\coloneqq k$ is its length.
We denote by $\run$ the set of all runs, and by $\runs{w}$ the set of
all runs starting with a finite path $w$.
Closing the set $\{\runs{w} \mid \text{$w$ is a finite path}\}$
under complements and countable unions generates
the standard Borel $\sigma$-algebra of measurable sets of runs.
Note that the functions $\CState{n}$ and $\Cnt{n}$ have measurable pre-images.

\begin{sloppypar}
A \emph{strategy} for player Max is a function,
$\sigma$, which to each finite path
\(
w
=
\conf{q_0}{c_0}
\cdots
\conf{q_k}{c_k}
,
\)
also called \emph{history} in
this context, where $q_k\in \sMax{Q}$, assigns a
probability distribution on the set of
rules of the form $(q_k,a,r)\in\rules$.
It is called \emph{pure} if $\sigma(w)$
assigns probability~$1$ to some transition, for each history $w$. 
We call $\sigma$ \emph{\stratCL{}} if $\sigma(w)$ depends only on the
last control state, $q_k$.
Strategies for Min are defined similarly, just by substituting
$\sMax{Q}$ with $\sMin{Q}$.
\end{sloppypar}

Assume that a pair $(\sigma,\pi)$ of strategies for Max and Min, respectively,
is fixed. Consider a finite path
\(
w
=
\conf{q_0}{c_0}
\cdots
\conf{q_k}{c_k}
\)
and a rule $(q_{i-1},c_i-c_{i-1},q_i)\in\rules$, $1\leq i\leq k$.
We assign to this rule a weight, $x_i$, as follows:
If $q_{i-1}\in\sP{Q}$ then $x_i=\Prob((q_{i-1},c_i-c_{i-1},q_i))$.
If $q_{i-1}\in\sMax{Q}$ then $x_i$ is equal to the probability of
$(q_{i-1},c_i-c_{i-1},q_i)$ assigned by 
$\sigma(
\conf{q_0}{c_0}
\cdots
\conf{q_{i-1}}{c_{i-1}})$,
and similarly for
$q_{i-1}\in\sMin{Q}$
and $\pi$.
The weight of $w$ is then $x_w=\prod_{i=1}^{\len{w}-1}x_i$, where the empty
product is equal to~$1$.
Once we also fix an {\em initial configuration}, $\conf{q}{c}$,
we obtain a probability measure $\Prb{\sigma,\pi}{\conf{q}{c}}$.
This is defined by setting
$\Pr{\sigma,\pi}{\conf{q}{c}}{\run(w)}=0$ if $w$ does not start with $\conf{q}{c}$, and
$\Pr{\sigma,\pi}{\conf{q}{c}}{\run(w)}=x_w$ if $w$ starts with $\conf{q}{c}$.
This and the requirement of countable additivity of a measure
already uniquely describes $\Prb{\sigma,\pi}{\conf{q}{c}}$ 
(see, e.g., \cite[p.~30]{Puterman94} 
for the case of MDPs.  The extension of this to SSGs is straightforward.)
If $\A$ is a maximizing OC-MDP, a minimizing OC-MDP,
or a OC-MC, we denote the probability measure by
$\Prb{\sigma}{\conf{q}{c}}$,
$\Prb{\pi}{\conf{q}{c}}$, or
$\Prb{}{\conf{q}{c}}$, respectively.

\paragraph{Objectives}
In this paper, an \emph{objective} for an OC-SSG is
a measurable set of runs.
Player Max tries to maximize the probability of this
set, whereas player Min tries to minimize it.
Given an objective, $O$, for a OC-SSG, $\A$, and a configuration, 
$s = \conf{q}{c}$,
we define the \emph{value in $s$} as
\[
\valO[\A]{O}{s}  \coloneqq  \sup_{\sigma} \inf_{\pi} \Pr{\sigma,\pi}{s}{O}
= \inf_{\pi} \sup_{\sigma} \Pr{\sigma,\pi}{s}{O}
.
\]
The latter equality follows from Martin's Blackwell determinacy 
theorem~\cite{M98}. We write just  $\valO{O}{s}$ if
$\A$ is understood.
For an $\eps\geq 0$, a strategy $\sigma$ for Max is called \emph{$\eps$-optimal in $s$}
if $\Pr{\sigma,\pi}{s}{O} \geq \valO{O}{s}-\eps$ for every $\pi$.
Similarly a strategy $\pi$ for Min is \emph{$\eps$-optimal in $s$}
if $\Pr{\sigma,\pi}{s}{O} \leq \valO{O}{s}+\eps$ for every $\sigma$.
A $0$-optimal strategy is called \emph{optimal}.
Note that by determinacy both players have $\eps$-optimal strategies  
for every $\eps >0$.

The key objective is the \emph{termination} objective:
\[
\Term \coloneqq
\{\omega \in \run \mid 
\exists n: \Cnt{n}(\omega) \leq 0 \}
.
\]
The name ``termination'' stems from the connection to one-counter
automata. Such automata also have a finite number of control states and
a non-negative counter, and a run can be considered to ``terminate'' 
upon hitting counter value $0$.
OC-SSGs do not necessarily halt when the counter is $0$,
and allow negative counter values. However, this difference is
irrelevant from the perspective of the termination objective,
for which only the part of runs with non-negative counter values
matter.

\paragraph{Games without a Counter}
In our arguments we also use the notion of
Simple Stochastic Games (SSGs) of Condon~\cite{C92}, which are similar 
to OC-SSGs.
The main difference is the lack of a counter, and the focus
on the objective of reaching a distinguished sink state.

\begin{definition}
A \emph{simple stochastic game (SSG)}
is a tuple
$\G=\ssgG$, where
\begin{itemize}
\item
$S$ is a
finite
set of \emph{states}, partitioned
into the states $\sMax{S}$ of player Max,
$\sMin{S}$ of player Min,
and stochastic states $\sP{S}$;
\item
$\btran{}\subseteq S\times S$
is a transition relation 
such that for every state $s\in S$ there is at least one state 
$r\in S$ such that $s\btran{} r$;  
\item
$\Prob$ is a \emph{probability assignment} which
to each $s \in \sP{S}$
assigns a rational probability distribution on
its set of successors, where for a state $s \in \sP{S}$
its successors are defined to be the set $\{ r \mid s \btran{} r \}$.
\end{itemize}
If $\sMin{S}=\emptyset$ we call $\G$
a \emph{maximizing Markov decision process (maximizing MDP)}.
If $\sMax{S}=\emptyset$ we call it a \emph{minimizing MDP}.
If $\sMax{S}=\sMin{S}=\emptyset$ we call $\G$ a
\emph{Markov chain}.

The SSG also comes with a distinguished sink state $s_0\in S$,
and this implicitly defines the {\em reachability objective}
``$\SSGReach$'' defined  by runs $\omega$
which visit $s_0$.
\end{definition}

Runs, strategies, probability measures and values with respect to
objectives are defined analogously to those for OC-SSGs,
just by removing references to the counter.
In particular, runs are sequences of states.
The following is well known.

\begin{fact}
\label{fact:reach}
(See, e.g., \cite{Puterman94,C92,CY98}.)
For both maximizing and minimizing MDPs,
optimal pure memoryless strategies for reachability exist and can be computed,
together with the optimal reachability value, in polynomial time.
\end{fact}

\section{Main Result}
\label{sec:res}

\begin{theorem}[Main]
\label{thm:main}

There is an algorithm that, 
given an OC-SSG, $\A$, a configuration, $\conf{q}{i}$, $i\geq 0$, and
a rational $\eps >0$, computes a rational
number, $\nu$,
such that $|\vt{q}{i}-\nu|\leq\eps$, and 
computes strategies
$\sigma$ and $\pi$ for the Max and Min player, respectively,
such that both $\sigma$ and $\pi$ are $\eps$-optimal starting in
$\conf{q}{i}$ with respect to the termination objective. 
The algorithm runs in nondeterministic time exponential
in $\size\A$ and polynomial in $\log(i)$ and $\log(1/\varepsilon)$. If
$\A$ is an OC-MDP, then the algorithm runs in deterministic time
exponential in $\size\A$ and polynomial in 
$\log(i)$ and $\log(1/\varepsilon)$.\footnote{\label{foot:nondet-explain}
To make precise the meaning of this theorem,
we have to spell out precisely what we mean by a {\em nondeterministic}
algorithm that computes $\nu$, $\sigma$ and $\tau$ within given resource 
bounds.
This is a standard notion for total search problems.
We will say a nondeterministic algorithm (i.e., nondeterministic Turing
machine) 
computes  $\nu$, $\sigma$ and $\tau$
within the specified resource bounds (namely, exponential time) 
if given any input the algorithms halts 
in exponential time on all computation paths, and 
furthermore, if the input is not well-formed
the algorithm ``rejects'' it on all computation paths,
but if the input is well-formed
(i.e., if it is given a well formed OC-SSG, $\A$, initial 
configuration $(q,i)$, and $\epsilon > 0$)
the nondeterministic algorithm:
(a)
 has at least one accepting
computation path; and
(b)  on every accepting computation path
it outputs values $\nu$, $\sigma$, and $\pi$ which satisfy
that  $|\vt{q}{i}-\nu|\leq \eps$ and such that $\sigma$ and
$\pi$ are $\epsilon$-optimal strategies for Max and Min, respectively,
for the given input OC-SSG $\A$ and initial configuration $(q,i)$. 
On rejecting computation paths the algorithm need not output anything.
Note that the outputs on different
accepting executions may be different, but they must 
all satisfy the required specification.}

\end{theorem}
\setcounter{theorem}{\value{repthmcnt}}

Let us first briefly sketch the main ideas in the proof of 
Theorem~\ref{thm:main}.
First, observe that for all $q \in Q$ and $i \leq j$
we have that $\vt{q}{i} \geq \vt{q}{j}\geq 0$. Let
\[
\mu_q \coloneqq \lim_{i \rightarrow \infty}  \vt{q}{i}
.
\]
Since $\mu_q \leq \vt{q}{i}$ for an arbitrarily large $i$, 
Player Max should be able to decrease the counter by an arbitrary 
value with probability at least $\mu_q$, no matter what Player Min does. 
The objective of ``decreasing the counter by an arbitrary value''
can be formalized directly as the following ``limit'' objective,
which has useful connections to termination~\cite{BBEKW10}:
\[
\CN \coloneqq
\{\omega \in \run \mid 
\liminf_{n\to\infty} \Cnt{n}(\omega) =-\infty \}
.
\]
OC-SSG with this objective are determined, which means that the following
value is defined for every $q \in Q$:
\begin{equation}
\label{eq:nu}
\nu_q
\coloneqq
\valO\CN{\conf{q}{n}},
\quad\text{where $n=0$}
.
\end{equation}
\begin{remark}
\label{rem:q}
Observe that due to the nature of $\CN{}$ we would obtain the same
value $\nu_q$ if we used any other value of $n$.
It will be often the case that we will measure the (optimal) probability
of some events, where the resulting number will not depend on the initial
counter value.
From now on, in such cases we will specify only the initial state, so, e.g.,
(1) would become
\(
\nu_q
\coloneqq
\valO\CN{q}
.
\)
\end{remark}
One intuitively expects that $\mu_q = \nu_q$, and we show that this is 
indeed the case (see Corollary~\ref{cor:term-lim}).
Further, by~\cite[Theorem~2]{BBE10}, $\nu_q$ is rational 
and computable in non-deterministic time polynomial in~$\size\A$. 
Moreover, both players have optimal pure \stratCL{} strategies 
$(\sigma^*,\pi^*)$ computable in non-deterministic polynomial time. 
For OC-MDPs, both the value $\nu_q$ and the optimal strategies
can be computed in deterministic time polynomial in $\size\A$.

Obviously, there must be a (sufficiently large)
$N$ such that $\vt{q}{i}-\mu_q \leq \eps$ for all $q \in Q$ and $i \geq N$.
We show that an upper bound on $N$ is computable, and is at most exponential 
in $\size\A$ and polynomial in $\log(1/\varepsilon)$, in
Section~\ref{sec-OC-MDP}. As we shall see, 
this part 
is highly non-trivial. For all configurations $\conf{q}{i}$, where $i \geq N$,
the value $\vt{q}{i}$ can be approximated by $\mu_q$ ($= \nu_q$), 
and both players can use the optimal strategies 
$(\sigma^*,\pi^*)$ for the $\CN$ objective. For the remaining
configurations $\conf{q}{i}$, where $i < N$, we consider a (finite-state)
SSG $\G$ obtained by restricting ourselves to configurations with counter 
between $0$ and $N$, extended by two fresh stochastic states 
$s_0,s_1$ with self-loops. All configurations of the form $\conf{q}{0}$ have
only one outgoing edge leading to $s_0$, and all configurations of the 
form $\conf{q}{N}$ can enter either $s_0$ with probability $\nu_q$, or
$s_1$ with probability $1 - \nu_q$.  In this SSG,
we compute the values and optimal strategies for the objective of reaching
$s_0$. This can be
done in nondeterministic time polynomial in the size
of $\G$ (i.e., exponential in $\size\A$). 
If $\A$ is an OC-MDP, then $\G$ is a MDP, and the values and optimal strategies 
can be computed in deterministic polynomial time in 
the size of $\G$ (i.e., exponential in $\size\A$) by linear
programming (this applies both to the ``maximizing'' and the
``minimizing''  OC-MDPs). Thus, we obtain the required approximations of  
$\vt{q}{i}$ for $i < N$, and the associated $\eps$-optimal strategies.

\begin{proof}[Proof of Theorem~\ref{thm:main}]
The algorithm is given an OC-SSG $\A=\ocssgA$, an initial configuration $\conf{q}{i}$, and
a rational number $\eps>0$, as input.
Recall that for $r\in Q$ we set
\(
\nu_r
\coloneqq
\valO\CN{r}
.
\)
The algorithm does the following:
\begin{enumerate}
\item\label{alg:CN} Compute a pair $(\sigma^*,\pi^*)$ of pure \stratCL{} strategies, for players Max and Min, respectively,
which are optimal for $\CN$ starting 
at every state $r\in Q$.  Compute $\nu_r$, for every $r \in Q$.
\item\label{alg:N} Compute $N$ such that $\vt{r}{j}-\nu_r \leq \eps$ for all $r \in Q$ and $j \geq N$.
\item If $i\geq N$ then return $\nu_q$, $\sigma^*$, $\pi^*$.
\item\label{alg:fin} Otherwise apply the algorithm from Lemma~\ref{lem:fin-seg}
to $\A$, $\eps$, $(\nu_r)_{r\in Q}$, $N$, $\sigma^*$, $\pi^*$ and return
$\nu_{\conf{q}{i}}$, $\bar\sigma$, $\bar\pi$ from its output.
\end{enumerate}

A key step is obviously step \ref{alg:fin}, which 
is not described here.  We shall
describe and prove the correctness and complexity
of that step in Lemma~\ref{lem:fin-seg}.
If we can carry out the computations as specified
in Steps \ref{alg:CN} and \ref{alg:N}, then the correctness
of the output in Step 3 holds by definition.
Let us now evaluate the complexity of the first two steps
(using some of our earlier results, and some results that will be established in
Section \ref{sec-OC-MDP}):

\begin{itemize}
\item (Step~\ref{alg:CN}.)
The values $\nu_r$, $r\in Q$, which have polynomially big encoding
by~\cite[Proposition~9]{BBE10}, can be guessed and verified in polynomial
time by~\cite[Theorem~2]{BBE10}.  Strategies 
exist that are optimal with respect to $\CN$
and are pure and \stratCL{} by~\cite[Proposition~7]{BBE10}.
We can guess such a strategy $\sigma^*$ and verify, using the numbers 
$\nu_r$, that it
is $\CN$-optimal for Max; similarly for $\pi^*$ and Min.
If $\A$ is an OC-MDP, all the above can be computed deterministically
in time polynomial in $\size{A}$.

\item (Step~\ref{alg:N}.)
Fixing $\pi^*$ in $\A$ we obtain a maximizing OC-MDP
$\A^*$. Lemma~\ref{lem:c-max} applied to $\A^*$
allows us to compute deterministically
a bound $N\in \exp(\size{\A'}^{O(1)})\cdot O(\log(1/\eps))$ such that in $\A^*$,
$\vt{r}{j}-\nu_r \leq \eps$ for all $r \in Q$ and $j \geq N$.
By Lemma~\ref{lem:ssg} this $N$ satisfies the requirements
of step~\ref{alg:N}.
\end{itemize}
\end{proof}

\subsection{Bounding counter value $N$ for  maximizing OC-MDPs}
\label{sec-OC-MDP}
Consider a maximizing OC-MDP
$\A = \ocmdpA$. 
Recall the definition of $\nu_q$ from (\ref{eq:nu}) and the notational
convention introduced in Remark~\ref{rem:q}. Specifically, 
we have $\nu_q = \sup_\sigma\Pr\sigma{q}{\CN}$
for all $q\in Q$.  Given $\eps > 0$, we show here how to
obtain a computable (exponential) bound on a number $N$ such that \(
\left| \vt{q}{i} -
\nu_q
\right|
<
\eps
\)
for all $i\geq N$.
We denote by $T$ the set of all states $q$ with $\nu_q=1$, and we
define the objective of reaching~$T$ as follows:
\[
  \ReachT \coloneqq
    \{\omega \in \run \mid \CState{i}(\omega) \in T 
    \mbox{ for some } i \geq 0\}.
\]
Further, we define the objective $\neg\ReachT \coloneqq \run \smallsetminus
\ReachT$.

\begin{fact}[\protect{cf.\ \cite[Proposition~3.2]{BBEKW10}}]
\label{fact:CN} The number $\nu_q$ is the maximal probability
of reaching $T$ from $q$ (see Remark~\ref{rem:q}), i.e., 
\[
\nu_q
= 
\valO\ReachT{q}
=
\sup_\sigma \Pr\sigma{q}{\ReachT}
.
\]
\end{fact}

\begin{lemma}
\label{lem:gap}
For all $q\in Q$ and $i\geq 0$
\begin{equation}
\label{eq:gap}
\nu_q
\leq
\vt{q}{i}
\leq
\sup_\sigma \Pr\sigma{\conf{q}{i}}{\Term \cap \neg\ReachT}
+
\nu_q
.
\end{equation}
\end{lemma}
\begin{proof}
The first inequality is obvious.
Because $\CN\cap\runs{\conf{q}{i}}\subseteq\Term\cap\runs{\conf{q}{i}}$,
we have $\vt{q}{i}=1$ for all $q\in T$, $i\geq 0$,
from which the second inequality follows by an easy application of 
the union bound.  Namely, for under strategy $\sigma$, the event
of termination can be split into the event of terminating
and not reaching $T$ unioned with the event of terminating and reaching $T$.
The probability of the latter event is clearly upper bounded by $v_q$.
\myqed\end{proof}

To provide the promised bound on $N$ we will prove an upper bound on
\(
\sup_\sigma \Pr\sigma{\conf{q}{i}}{\Term \cap \neg\ReachT}
\)
which decreases toward $0$ exponentially fast in $i$.
We will first define a suitably restricted class
of OC-MDPs (we call them ``rising'' OC-MDPs)
and find such a bound using martingale theory (Lemma~\ref{lem:c-max-no-T})
for that restricted class.
We then extend the results to all OC-MDPs (Lemma~\ref{lem:c-max})
by showing that for every OC-MDP, $\A$, there is a polynomially bigger 
``rising'' OC-MDP, $\bar\A$,
which ``embedds'' in it the states of the original OC-MDP,
and preserves the rate with which
\(
\sup_\sigma \Pr\sigma{\conf{q}{i}}{\Term \cap \neg\ReachT}
\)
reaches $0$ from those corresponding states.  We will make this precise later.

To be able to use the martingale theory methods for rising OC-MDPs
we need to guarantee that in each rising OC-MDP, under 
every pure \stratCL{} strategy,
\(
\liminf_{i\to\infty}
\Cnt{i}/i
\)
is almost surely positive.
This value is sometimes called the mean-payoff,
see also~\cite{BBEKW10}.
We now state the definition of a rising OC-MDP using two simple properties,
and show that these two properties guarantee that the mean-payoff is almost surely
positive.

\begin{definition}
\label{def:idling}
A pure \stratCL{} strategy, $\sigma$, is \emph{idling},
if there is a state $q\in Q$, such that
\(
\Pr\sigma{\conf{q}{0}}{\exists i>0:\CState{i}=q}=1
\)
and for all $i\geq 0$:
\(
\Pr\sigma{\conf{q}{0}}{\CState{i}=q \implies \Cnt{i}=0}=1
.
\)

A maximizing OC-MDP is called \emph{rising}
if 
$T= \{q\in Q \mid \nu_q=1\}=\emptyset$ and
no pure \stratCL{} strategy is idling.
\end{definition}

\begin{lemma}
\label{lem:liminf}
Let $\A=\ocmdpA$ be a rising OC-MDP.
Then for every pure \stratCL{} strategy, $\sigma$,
and every $q\in Q$ we have
\(
\Pr\sigma{q}{\liminf_{i\to\infty}\Cnt{i}/i > 0}  =  1.
\)
\end{lemma}

\begin{proof}
Let us fix a pure \stratCL{} strategy $\sigma$.
Because $\sigma$ is \stratCL{}, there is a collection of disjoint subsets of $Q$,
called ergodic sets, or bottom strongly connected components (BSCCs),
in the standard theory of Markov chains,
such that almost all runs end up visiting infinitely often exactly
the states of some of the BSCCs.
Let us focus, for a while, on a single BSCC $C\subseteq Q$.
By standard results, for each pair of states $r,s\in C$ the play from $r$ almost surely visits
$s$, and the expected time to visit $s$ from $r$ is finite.
As a consequence, there is a unique constant $p$
such that
\(
\Pr\sigma{r}{\liminf_{i\to\infty}\Cnt{i}/i > 0}  =  p
\)
for all $r\in C$,
because $\liminf_{i\to\infty}\Cnt{i}/i > 0$ is a prefix-independent
property. Moreover, we can use the result from \cite[Theorem~3.2]{GH-SODA10}
which says that in a presence of $r\in C$ such that
\(
\Pr\sigma{r}{\liminf_{i\to\infty}\Cnt{i}/i > 0}  >  0
\)
there must also be some $s\in C$ such that
\(
\Pr\sigma{s}{\liminf_{i\to\infty}\Cnt{i}/i > 0}  =  1.
\)
Thus either $p=0$ or $p=1$.
Let us first prove by contradiction that $p=1$, and then we shall
consider the more general case 
where rather than assuming $r \in C$ for a BSCC $C$,
we consider an arbitrary start state in the entire OC-MDP.

Assume that $p=0$.
By Fact~\ref{fact:CN}, $T=\emptyset$ implies that
\(
\Pr\sigma{r}{\liminf_{i\to\infty}\Cnt{i} = -\infty}  =  0
\)
for all $r\in C$.
It is easy to see that this implies that
\(
\Pr\sigma{r}{\liminf_{i\to\infty}\Cnt{i}/i \geq 0}  =  1
\)
for all $r\in C$ and all strategies $\sigma$.
Due to our assumption of $p=0$,
\(
\Pr\sigma{r}{\liminf_{i\to\infty}\Cnt{i}/i = 0}  =  1
\)
for all $r\in C$.
Now Lemma~3.3 from~\cite{BBEKW10} says:
``For all $q$, the pure \stratCL{} strategies $\tau$ which satisfy
\(
\Pr\tau{q}{\liminf_{i\to\infty}\Cnt{i} = -\infty}  =  1
\)
are exactly those which satisfy
\(
\Pr\tau{q}{\liminf_{i\to\infty}\Cnt{i}/i \leq 0}  =  1
\)
and
\(
\Pr\tau{\conf{q}{0}}{\exists i: \Cnt{i} < 0}  >  0.
\)''
But we do not have any strategies of the first kind,
so there is a state $r\in C$ such that
\(
\Pr\sigma{\conf{r}{0}}{\forall i: \Cnt{i} \geq 0}  =  1.
\)
If
\(
\Pr\sigma{\conf{r}{0}}{\exists i>0: \CState{i}=r \land \Cnt{i} > 0}  >  0
\)
then because the expected time between two visits to $r$ is finite,
it can fairly easily be established that
\(
\Pr\sigma{r}{\liminf_{i\to\infty}\Cnt{i}/i > 0}  >  0,
\)
which would contradict the assumption $p=0$.
Thus
\(
\Pr\sigma{\conf{r}{0}}{\forall i>0: \CState{i}=r \implies \Cnt{i} = 0}  =  1
\)
and $\sigma$ is thus idling. But this is not possible because 
by assumption $\A$ is rising,
so the assumption of $p=0$ cannot be satisfied, and we have proved that $p=1$.

Now consider an arbitrary state $q\in Q$.
Because $\liminf_{i\to\infty}\Cnt{i}/i > 0$ is prefix-independent,
and almost every run from $q$ reaches some BSCC, where
$\liminf_{i\to\infty}\Cnt{i}/i > 0$ is satisfied almost surely, we have
\(
\Pr\sigma{q}{\liminf_{i\to\infty}\Cnt{i}/i > 0}  =  1.
\)
\end{proof}

Now we define a suitable submartingale for a given rising
OC-MDP, and use Azuma's inequality to show that 
$\sup_\sigma \Pr\sigma{\conf{q}{i}}{\Term \cap \neg\ReachT}$
decreases to~$0$ exponentially fast in $i$. Recall that a stochastic process
$\ms{0},\ms{1},\dots$ is a submartingale if, for all $i \geq 0$,
$\Ex{}{}{|\ms{i}|} < \infty$, and \mbox{$\Ex{}{}{\ms{i+1} \mid
  \ms{1},\dots,\ms{i}} \geq \ms{i}$} almost surely. If we further assume
that $|\ms{i+1}-\ms{i}| \leq c$ almost surely for all $i \geq 0$, we can apply
the  \emph{Azuma-Hoeffding inequality}\footnote{In the
literature (see, e.g. \cite{GriSti92}), the
Azuma-Hoeffding inequality is usually stated for martingales
and supermartingales where is takes the form 
$\Pr{}{}{\ms{n}-\ms{0} \geq t} \ \leq\ \exp(-t^2/2nc^2)$. 
Inequality~(\ref{eq:azuma}) is obtained just by realizing that if
$\ms{0},\ms{1},\dots$ is a submartingale, then
$-\ms{0},-\ms{1},\dots$ is a supermartingale.}, which
says that the following holds for all $t>0$ and $n\geq 0$:
\begin{equation}
  \label{eq:azuma}
  \Pr{}{}{\ms{n}-\ms{0} \leq t} \quad \leq \quad
  \exp\left(\frac{-t^2}{2nc^2}\right)
\end{equation}

Let $\A=\ocmdpA$ be a rising OC-MDP. Since $\A$ is rising, 
the mean-payoff (i.e.,
the average change of the counter per transition) is almost surely
positive for all pure counterless strategies. Since there are only finitely
many pure counterless strategies, there is even a fixed bound $x > 0$ such 
that the mean payoff is larger than~$x$ almost surely. This means that after
performing $i$~transitions, the counter should increase at least by 
$i\cdot x$ on average. Hence, one might be tempted to define
\[
\mar{i}
\coloneqq
\begin{cases}
\Cnt{i} - i\cdot x
&
\text{if $\Cnt{j}>0$ for all $j,\ 0\leq j< i$,}\\
\mar{i-1} & \text{otherwise.}
\end{cases}
\]
and try to prove that $\mar{0},\mar{1},\ldots$ is a submartingale. 
Unfortunately, this does not work, because some control states may not 
allow to increase the counter by~$x$ or more. A similar problem was encountered
previously in \cite{BKK11} in the context of purely probabilistic OC automata,
and the difficulty was overcome by employing ``artificial'' additive
constants that compensate the difference among the individual control 
states. We show that a similar trick works also in our setting. That is,
we aim at designing a submartingale of the following form: 
\[
\mar{i}
\coloneqq
\begin{cases}
\Cnt{i} +
\poten_{\CState{i}} - i\cdot x
&
\text{if $\Cnt{j}>0$ for all $j,\ 0\leq j< i$,}\\
\mar{i-1} & \text{otherwise.}
\end{cases}
\]
Here $\poten_q$ is a suitable constant that depends only on~$q$. 
However, it is not clear whether the constants $\poten_q$ can 
be chosen so that $\ms{0},\ms{1},\dots$ becomes a submartingale,
and what is the size of these constants if they exist. This problem
is solved simply by observing that the defining property of a 
submartingale (see above) immediately gives a system of linear
inequality constraints that should be satisfied by $\poten_q$. For example,
suppose that $\Cnt{i} = j$ and $\CState{i} = q$ where $q \in \sP{Q}$.
For \emph{every} Max strategy, we would like to have that 
\mbox{$\Ex{}{}{\ms{i+1} \mid \ms{i}} \geq \ms{i}$}. This means to ensure
that this inequality is satisfied for every outgoing transition 
of $\conf{q}{j}$, i.e., for every $(q,k,r) \in \Delta$
we wish to have
\[
  \Ex{}{}{\ms{i+1} \mid \ms{i}} \quad = \quad 
  (j+k) + \poten_r - (i+1)\cdot\bar{x} \quad \geq \quad 
  \ms{i} \quad = \quad j + \poten_q - i \cdot\bar{x}.
\]
This yields $\poten_q \leq -x + k + \poten_r$. Note that if $q$ is 
stochastic, we need to consider the ``weighted sum'' of the  
outgoing transition of $\conf{q}{j}$ instead. Thus, we obtain the
system of linear inequalities of Figure~\ref{fig:system-L}.

Now we show that the linear system of inequalities given  in
 Figure~\ref{fig:system-L} 
has a non-negative rational solution, and derive a bound on its size.
Then, we take this solution, define the associated submartingale,
and use Azuma's inequality to derive the desired result. 

\begin{figure}[t]
\begin{align*}
\poten_q
&\leq
-x
+
k
+
\poten_r
&
\text{for all $q\in\sMax{Q}$ and $(q,k,r)\in\rules$,}
\\
\poten_q
&\leq
-x
+
\textstyle\sum_{(q,k,r)\in\rules}
P((q,k,r))\cdot
(k+\poten_r)
&
\text{for all $q\in\sP{Q}$,}
\\
x &> 0.
\end{align*}
\caption{The system $\sL$ of linear inequalities over $x$ and $\poten_q$, $q\in Q$.}
\label{fig:system-L}
\end{figure}

\begin{lemma}
\label{lem:sol}
Let $\A=\ocmdpA$ be a rising OC-MDP.
Then there is a non-negative
rational solution $(\bar{x},(\bar{\poten}_q)_{q\in Q})\in \Qset^{|Q|+1}$ to $\sL$,
such that $\bar{x}>0$.
(The binary encoding size of the solution is polynomial
in $\size{\A}$.)
\end{lemma}

\begin{proof}
We first prove that there is some non-negative solution to $\sL$ with $\bar{x}>0$.
The bound on size then follows by standard facts about
linear programming.
To find a solution, we will use optimal values for minimizing
\emph{discounted total reward} in $\A$.
For every discount factor, $\lambda$, $0<\lambda<1$, and a strategy, $\tau$,
we denote the discounted total reward, starting under $\tau$, by
\(
\dtr\lambda{q}{\tau}
\coloneqq
\sum_{i\geq 0}
\lambda^i\cdot \Ex\tau{q}{\Cnt{i+1}-\Cnt{i}}
\),
and set $\mindtr\lambda{q} \coloneqq \inf_\tau\dtr\lambda{q}\tau$.
We prove that there is some $\lambda$, such that setting
\(
\bar{\poten}_q\coloneqq \mindtr\lambda{q}
\)
and
\begin{multline*}
\bar{x} \coloneqq
\min
\big(
\{
k
+
\mindtr\lambda{r}
-
\mindtr\lambda{q}
\mid
q \in \sMax{Q},
(q,k,r)\in\rules
\}
\\
\cup
\{
P((q,k,r))\cdot
\left(
k +
\mindtr\lambda{r}
-
\mindtr\lambda{q}
\right)
\mid
q \in \sP{Q},
(q,k,r)\in\rules
\}
\big)
\end{multline*}
forms a non-negative solution to $\sL$ with $\bar{x}>0$.

\begin{sloppypar}
Now we proceed in more detail.
First we choose the right $\lambda$.
Lemma~\ref{lem:liminf} and our assumptions guarantee that
\(
\Pr{\tau}{q}{
\liminf_{i\to\infty}
\Cnt{i}/i > 0
}
=
1
\)
for every pure \stratCL{} strategy $\tau$.
Thus
\(
\sum_{i\geq 0}
\cdot \Ex\tau{q}{\Cnt{i+1}-\Cnt{i}}
=
\infty
,
\)
and hence for every such $\tau$ there is a $\Lambda_\tau<1$ such that
\(
\dtr\lambda{q}{\tau} > 0
\)
for all $q\in Q$ and $\lambda\geq\Lambda_\tau$.
There are only finitely many pure \stratCL{} strategies, and we choose
our $\lambda$ to be $\lambda \coloneqq \max_{\tau}\Lambda_\tau<1$.
\end{sloppypar}

Having fixed the $\lambda$ above, we now 
prove that there is a pure \stratCL{} strategy $\sigma$, such
that $\mindtr\lambda{q}=\dtr\lambda{q}\sigma$ for all $q$.
By standard results (e.g.,~\cite{Puterman94}), translated from the terminology of MDPs
with rewards to that of OC-MDPs, for a fixed state, $q$,
there is always a pure \stratCL{}
strategy $\sigma_q$
such that
\(
\dtr\lambda{q}{\sigma_q}
=\mindtr\lambda{q}
.
\)
Moreover, this strategy has to play optimally in successors of $q$ as well,
thus there is in fact a single pure \stratCL{} strategy $\sigma$
such that for all $q$:
\(
\dtr\lambda{q}{\sigma}
=\mindtr\lambda{q}
.
\)

Finally, $\bar{x}>0$, because for all $q\in \sP{Q}$
\begin{align*}
\dtr\lambda{q}\sigma
&=
\sum_{i\geq 0}
\lambda^i\cdot \Ex\sigma{q}{\Cnt{i+1}-\Cnt{i}}
\\
&=
\sum_{(q,k,r)\in\rules}
P((q,k,r))\cdot
\left(
k +
\lambda\cdot
\sum_{i\geq 0}
\lambda^i\cdot \Ex\sigma{r}{\Cnt{i+1}-\Cnt{i}}
\right)
\\
&=
\sum_{(q,k,r)\in\rules}
P((q,k,r))\cdot
\left(
k +
\lambda\cdot
\dtr\lambda{r}\sigma
\right)
\\
&<
\sum_{(q,k,r)\in\rules}
P((q,k,r))\cdot
\left(
k +
\dtr\lambda{r}\sigma
\right)
,
\end{align*}
the last inequality following from $\dtr\lambda{r}\sigma>0$ for all
$r\in Q$; and similarly for all $q\in \sMax{Q}$ and $(q,k,r)\in\rules$
\begin{multline*}
\dtr\lambda{q}\sigma
=
\sum_{i\geq 0}
\lambda^i\cdot \Ex\sigma{q}{\Cnt{i+1}-\Cnt{i}}
\\
\leq
k +
\lambda\cdot
\sum_{i\geq 0}
\lambda^i\cdot \Ex\sigma{r}{\Cnt{i+1}-\Cnt{i}}
=
k +
\lambda\cdot
\dtr\lambda{r}\sigma
<
k +
\dtr\lambda{r}\sigma
.
\end{multline*}
Here the first inequality follows from the fact that $\sigma$ minimizes the
discounted total reward.
\myqed\end{proof}

Given the solution
$(\bar{x},(\bar{\poten}_q)_{q\in Q})\in \Qset^{|Q|+1}$
from Lemma~\ref{lem:sol},
we define a sequence of random variables
$\{\mar{i}\}_{i\geq0}$
by setting
\[
\mar{i}
\coloneqq
\begin{cases}
\Cnt{i} +
\bar{\poten}_{\CState{i}} - i\cdot \bar{x}
&
\text{if $\Cnt{j}>0$ for all $j,\ 0\leq j< i$,}\\
\mar{i-1} & \text{otherwise.}
\end{cases}
\]

We shall now show that $\mar{i}$ defines a submartingale.

\begin{lemma}
Let $\A=\ocmdpA$ be a rising OC-MDP and $\{\mar{i}\}_{i\geq0}$ defined as above.
Under an arbitrary strategy $\tau$ and
with an arbitrary initial configuration $\conf{q}{n}$,
the process $\{\mar{i}\}_{i\geq0}$ is a submartingale.
\end{lemma}

\begin{proof}
Consider a fixed path, $u$, of length $i\geq 0$.
For all $j,\ 0\leq j\leq i$
the values $\Cnt{j}(\omega)$ are the same
for all $\omega\in\runs{u}$.
We denote these common values by $\Cnt{j}(u)$,
and similarly for $\CState{j}(u)$ and $\mar{j}(u)$.
If $\Cnt{j}(u)=0$ for some $j\leq i$,
then $\mar{i+1}(\omega)=\mar{i}(\omega)$
for every $\omega \in \runs{u}$.
Thus $\Ex\tau{\conf{q}{n}}{\mar{i+1}\mid \runs{u}}=\mar{i}(u)$.
Otherwise,
consider the last configuration, $\conf{r}{l}$,
of $u$.
For every possible successor, $\conf{r'}{l'}$, set
\[
p_{\conf{r'}{l'}}
\coloneqq
\begin{cases}
\tau(u)(\conf{r}{l}\to\conf{r'}{l'})
&\text{if $r\in\sMax{Q}$,}\\
Prob(\conf{r}{l}\to\conf{r'}{l'})
&\text{if $r\in\sP{Q}$.}
\end{cases}
\]
Then
\[
\Ex\tau{\conf{q}{n}}{
\Cnt{i+1} -
\Cnt{i}
+
\bar{\poten}_{\CState{i+1}} - \bar{x}
\mid \runs{u}} %
=\quad 
- \bar{x}
+
\sum_{(r,k,r')\in\rules}
p_{\conf{r'}{l+k}}\cdot
(k+\bar{\poten}_{r'})
\quad \geq\quad \bar{\poten}_r
.
\]
This allows us to derive the following:
\begin{align*}
\Ex\tau{\conf{q}{n}}{\mar{i+1} \mid \runs{u}}
&=
\Ex\tau{\conf{q}{n}}{\Cnt{i+1} + \bar{\poten}_{\CState{i+1}} - (i+1)\cdot \bar{x}
\mid \runs{u}}
\\
&=
\Cnt{i}(u)
+
\Ex\tau{\conf{q}{n}}{
\Cnt{i+1} -
\Cnt{i}
+
\bar{\poten}_{\CState{i+1}} - \bar{x}
\mid \runs{u}}
- i\cdot \bar{x} 
\\
&\geq
\Cnt{i}(u)
+
\bar{\poten}_{\CState{i}(u)}
- i\cdot \bar{x} 
\quad
=
\quad
\mar{i}(u)
.
\end{align*}
\end{proof}

Now we have prepared all that we need to bound
\(
\sup_\sigma \Pr\sigma{\conf{q}{i}}{\Term}
\) for rising OC-MDPs.

\begin{lemma}
\label{lem:c-max-no-T}
Given a rising OC-MDP, $\A$, one can compute a rational constant
$c<1$, and an integer $h \geq 0$ such that for all $i\geq h$ and $q\in Q$ 
\[
\sup_\sigma \Pr\sigma{\conf{q}{i}}{\Term}
\leq \frac{ c^{i} }{1-c}.
\]
Moreover, $c \in \exp(1/2^{\size{\A}^{O(1)}})$ and $h \in \exp(\size{\A}^{O(1)})$.
\end{lemma}

\begin{proof}
Denote by $\Term_j$ the event of terminating
after \emph{exactly} $j$ steps.
Further set
\(
\vdiff
\coloneqq
\max_{q\in Q} \bar{\poten}_q
-
\min_{q\in Q} \bar{\poten}_q
,
\)
and assume that $\Cnt{0} \geq \vdiff$.
Then the event $\Term_j$ implies that
\(
\mar{j}-\mar{0}
=
\bar{\poten}_{\CState{j}} - j\cdot \bar{x}
-
\Cnt{0} -
\bar{\poten}_{\CState{0}}
\leq
- j\cdot \bar{x}.
\)
Finally, observe that we can bound the one-step change
of the submartingale value by
\(
\vdiff
+
\bar{x}
+
1
.
\)
Using the Azuma-Hoeffding inequality for the submartingale           
$\{\mar{n}\}_{n\geq 0}$  (see, e.g., Theorem 12.2.3 in \cite{GriSti92}), 
we thus obtain the following bound for every strategy
$\sigma$ and initial configuration $\conf{q}{i}$ with
$i\geq\vdiff$:
\[
\Pr\sigma{\conf{q}{i}}{\Term_j}
\leq
\Pr\sigma{\conf{q}{i}}{\mar{j}-\mar{0} \leq -j\cdot\bar{x}}
\leq
\exp
\left(
\frac%
{-\bar{x}^2\cdot j^2}%
{2j\cdot(\vdiff+\bar{x}+1)}
\right)
.
\]
We choose
\(
c
\coloneqq
\exp
\left(
\frac%
{-\bar{x}^2}%
{2\cdot(\vdiff+\bar{x}+1)}
\right)
<1
\)
and
$h\coloneqq \lceil \vdiff \rceil$,
and observe that
for all $q\in Q$, $i\geq h$:
\[
\Pr\sigma{\conf{q}{i}}{\Term}
=
\sum_{j\geq i}
\Pr\sigma{\conf{q}{i}}{\Term_j}
\leq
\sum_{j\geq i} c^j
=
\frac{c^i}{1-c}
.
\]
The given bounds on $c$ and $h$ are easy to check,
and the detailed computation can be found
\ifApp{in Section~\ref{sec:app-bounds}.}{in~\cite{fullversion}.}
\end{proof}

As a final step, we extend the results to the general case of (not necessarily rising) OC-MDPs.
\begin{lemma}
\label{lem:c-max}
Given a maximizing OC-MDP, $\A'=\OCMDP{Q'}{\rules'}{P'}$, one can compute a rational constant
$c<1$, and an integer $h \geq 0$ such that for all $i\geq h$ and $q\in
Q$ 
\[ \sup_\sigma \Pr\sigma{\conf{q}{i}}{\Term \cap \neg\ReachT}
\leq \frac{ c^{i} }{1-c}.
\]
Moreover, $c \in \exp(1/2^{\size{\A'}^{O(1)}})$ and $h \in
\exp(\size{\A'}^{O(1)})$.
As a consequence,
a number $N$ such that \(
\left| \vt{q}{i} -
\sup_\sigma\Pr\sigma{q}{\CN}
\right|
<
\eps
\)
for all $q\in Q'$ and $i\geq N$ satisfies
\(
N
\leq
\max\{
h,
\lceil
\log_c(\eps\cdot(1-c))
\rceil
\}
\in \exp(\size{\A'}^{O(1)})\cdot O(\log(1/\eps))
.
\)
\end{lemma}

\begin{proof}
The heart of the proof is a reduction which computes a polynomially bigger
rising OC-MDP $\bar\A=\OCMDP{\bar{Q}}{\bar\rules}{\bar{P}}$
from $\A'$, uses the algorithm from Lemma~\ref{lem:c-max-no-T}
to compute the bounds $c$ and $h$ for $\bar\A$, and returns the very same numbers for $\A'$.
The reduction itself is in two steps, first computing an OC-MDP $\A=\ocmdpA$
from $\A'$ such that $\valO\CN{q}<1$ for all $q\in Q$ in $\A$, and then
$\bar\A$ from $\A$.

The first step, from $\A'$ to $\A$ is easier.
Recall that we called $T=T(Q)$ the set of all $q\in Q$
such that $\valO\CN{q}=1$.
Here we use it also to denote the analogous subset
$T(Q')$
of $Q'$
of all states $q\in Q'$
such that $\valO\CN{q}=1$ in $\A'$.
Theorem~3.1 from~\cite{BBEKW10} guarantees that we
can compute the set $T(Q')$ in time polynomial in $\size{\A'}$.
Then we set $Q\coloneqq (Q'\setminus T(Q'))\cup\{\qd\}$, with $\sP{Q}=\sP{Q'}\setminus T(Q')$.
The state $\qd$ has a unique outgoing rule in $\rules$: $(\qd,+1,\qd)$.
The rest of the rules in $\rules$ are derived from $\rules'$ by redirecting
all rules ending in $T(Q')$ to $\qd$. $P'$ is derived from $\bar{P}$
accordingly.
It is easy to see that $T(Q)=\emptyset$, because
$T(Q')\cap Q=\emptyset$ and by construction, $T(Q)\subseteq T(Q')$.

A technique to achieve the second step was already partially developed in 
our previous work~\cite{BBEKW10}, where
we used the term ``decreasing'' for rising strategies.
There we gave a construction which preserves the property of
optimal termination probability being $=1$.
We in fact can establish that a similar construction preserves the 
exact termination value. 
Because the idea is not new, we leave details to
\ifApp{Section~\ref{sec:no-idling}.}{\cite{fullversion}.}
The important properties of $\bar\A$ are stated in the
following lemma, the proof of which can be found in
\ifApp{Section~\ref{sec:no-idling},}{\cite{fullversion},}
along with the formal definition of $\bar\A$.

\begin{lemma}
\label{lem:reduction}
There is a deterministic polynomial-time algorithm which given a maximizing
OC-MDP, $\A=\ocmdpA$, computes another maximizing OC-MDP,
$\bar\A=\OCMDP{\bar{Q}}{\bar\rules}{\bar{P}}$,
and a map $f:Q\to\bar{Q}$
satisfying:
\begin{itemize}
\item
$\size{\bar\A}\in\mathcal{O}(\size{\A}^4).$
\item
There are no idling pure \stratCL{} strategies in $\bar\A$.
\item
$\vt{q}{i}=\vt{f(q)}{i}$
for all $q\in Q$ and $i\geq 0$.
\item
If $\valO\CN{q}<1$ for all $q\in Q$ in $\A$, then $\bar\A$ is rising.
\end{itemize}
\end{lemma}

In particular, note that $\bar\A$ obtained from our $\A$ is rising.
Now let $q\in Q'$ be a state of $\A'$, such that
$q \notin T(Q')$.
We know that
\(
\sup_\sigma \Pr\sigma{\conf{q}{i}}{\Term \cap \neg\ReachT}
\) in $\A'$
equals
\(
\sup_\sigma \Pr\sigma{\conf{q}{i}}{\Term}
\) in $\A$, which in turn equals
\(
\sup_\sigma \Pr\sigma{\conf{f(q)}{i}}{\Term}
\)
in $\bar\A$.
Note that $\size{\bar\A}\in\size{\A'^{O(1)}}$.
Applying Lemma~\ref{lem:c-max-no-T} to $\bar\A$ finishes
the proof of the first part of Lemma~\ref{lem:c-max}.

In the second part the inequality
\(
N
\leq
\max\{
h,
\lceil
\log_c(\eps\cdot(1-c))
\rceil
\}
\)
is an easy computation. Verifying that
\(
\lceil
\log_c(\eps\cdot(1-c))
\rceil
\in \exp(\size{\A'}^{O(1)})
\)
is also easy\ifApp{ and can be found in Section~\ref{sec:app-bounds}.}{, see~\cite{fullversion}.}
\end{proof}

Also, as an immediate consequence of Lemma~\ref{lem:gap} and Lemma~\ref{lem:c-max}
we obtain the following:
\begin{corollary}
\label{cor:MDP-term-lim}
  For every $q \in Q$, $\nu_q = \lim_{i \rightarrow \infty} \vt{q}{i}$. 
\end{corollary}

\subsection{Bounding $N$ for general SSGs}
\label{sec:OC-SSG}

By~\cite[Proposition~7]{BBE10}, player Min always has an optimal pure \stratCL{}
strategy, $\opi$, such that
\[
   \valO\CN{q} =  \sup_\sigma \Pr{\sigma,\opi}{q}{\CN}.
\]
By fixing the choices of $\opi$ in $\A$
we obtain a maximizing OC-MDP, 
$\A^* = \OCMDP{Q^*}{\delta^*}{P^*}$, where
$\sP{Q^*}=\sP{Q}\cup\sMin{Q}$, $\sMax{Q^*}=\sMax{Q}$,
$\delta^* \coloneqq 
\{ (q,k,r) \in \delta \mid q \in \sP{Q}\cup\sMax{Q} \lor \opi(q)=r \}$,
and $P^*$ is the unique (for $\A^*$) extension of $P$ to states from
$\sMin{Q}$.

\begin{lemma}
\label{lem:ssg}
Let $\A=\ocssgA$ be an OC-SSG, $\opi$ a $\CN$-optimal strategy
for Min, and $\A^*$ the minimizing OC-MDP given by fixing $\opi$ in $\A$
as described above.
Then for all $q\in Q$:
\begin{gather}
   \label{eq:game-values}
   \forall i\geq0:
   \lim_{j \to \infty}  \vt[\A]{q}{j} \leq
   \vt[\A]{q}{i} \leq \vt[\A^*]{q}{i}.\\
\label{eq:game-values-lim}
  \lim_{j \rightarrow \infty}  \vt[\A]{q}{j} = 
   \lim_{j \rightarrow \infty}  \vt[\A^*]{q}{j}.
\end{gather}
\end{lemma}

\begin{proof}
Since for all $j$ we have
\(
\vt[\A]{q}{j+1} \leq \vt[\A]{q}{j},
\), we obtain the first inequality in (\ref{eq:game-values}).
The second inequality in (\ref{eq:game-values}) follows from the fact that in $\A^*$ we restricted the
possible moves of Min.
The ``$\leq$'' direction in (\ref{eq:game-values-lim}) follows directly
from~(\ref{eq:game-values}), and the other direction is obtained as follows: 
\[
 \begin{array}{lclcl}
  \lim_{i \rightarrow \infty}  \vt[\A]{q}{i} & \geq &
     \valO[\A]\CN{q} & \hspace*{1em} & \mbox{(immediate)}\\
  & = & \valO[\A^*]\CN{q} && \mbox{(immediate)}\\
  & = & \lim_{i \rightarrow \infty}  \vt[\A^*]{q}{i} &&
  \mbox{(by Corollary~\ref{cor:MDP-term-lim})}
 \end{array}
\]
\end{proof}

\begin{corollary}
\label{cor:term-lim}
  For every control state $q$ of an OC-SSG $\A$ we have that
  \[
    \lim_{i \rightarrow \infty}  \vt[\A]{q}{i} = \valO[\A]\CN{q}.
  \]
\end{corollary}

\subsection{Analyzing a Finite Segment of Configurations}
\label{sec:reach-analysis}

\begin{lemma}
\label{lem:fin-seg}
There is a nondeterministic algorithm\footnote{Again, see
footnote \ref{foot:nondet-explain} for a precise explanation
of what we mean by a nondeterministic algorithm in this context.} that given
an OC-SSG, $\A=\ocssgA$, and a rational $\eps>0$
as input, and, in addition, given the following
precomputed values:
\begin{itemize}
\item
$\nu_q=\valO{\CN}{q}$ for every $q\in Q$
\item
an integer $N\geq 0$
such that
\(
0\leq\vt{q}{i}-\nu_q\leq\eps
\)
for all $q\in Q$ and $i\geq N$,
\item
and a pair of strategies $(\sigma^*,\pi^*)$ for Max and Min
which are optimal for $\CN$ in all $q\in Q$;
\end{itemize}
computes the following output:
\begin{itemize}
\item
a number $\nu_{\conf{q}{i}}$
for each $q\in Q$ and $i\leq N$
such that
\(
0\leq\vt{q}{i}-\nu_{\conf{q}{i}}\leq\eps
,
\)
\item
and a pair of strategies $(\bar\sigma,\bar\pi)$ for Max and Min, respectively,
which are $\eps$-optimal for termination
in all configurations.
\end{itemize}
The algorithm runs in time polynomial in $N{\cdot}\size{\A}$.
Furthermore,
if $\A$ is an OC-MDP then the algorithm is deterministic.
\end{lemma}

\begin{proof}
The first idea is to analyze the following SSG, $\G$,
which is essentially $\A$ restricted to configurations
with counter value between $0$ and $N$.
The set of states of $\G$ is $\{ \conf{q}{i} \mid q\in Q, 0\leq i \leq N\}
\cup \{s_0, s_1\}$.
The ownership of the states of the form $\conf{q}{i}$, $0<i<N$ is the same
as in $\A$, the states $s_0, s_1$ and $\conf{q}{i}$ for $q\in Q$, $i\in\{0,N\}$ are stochastic.
For $0<i<N$, there is a transition $\conf{q}{i}\btran{}\conf{r}{j}$
iff $(q,j-i,r)\in\delta$.
Probabilities of these transitions, where applicable, are derived from $P$.
Vertices of the form $\conf{q}{0}$, and the state $s_0$ have only one transition,
to $s_0$.
Vertices of the form $\conf{q}{N}$ have transitions to both $s_0$
and $s_1$, and the state $s_1$ has only the self-loop transition.
The probability of a transition $\conf{q}{N}\btran{}s_0$
equals $\vt{q}{N}$ for all $q$.

Clearly we have
\(
\sup_\sigma
\inf_\pi
\Pr{\sigma,\pi}{\conf{q}{i}}{\SSGReach}
=
\vt{q}{i}
\)
for all $q\in Q$ and $i\leq N$.
The problem is that the transition probabilities from $\conf{q}{N}$
in $\G$ are unknown (and may even be irrational).   We will not actually construct $\G$.
To use such reachability
analysis for approximating the termination values we have to switch to
a slightly perturbed SSG, $\G'$.

$\G'$ is almost identical to $\G$: it has the same sets of states and
transitions. The only difference is that in $\G'$
the probability of a transition $\conf{q}{N}\btran{}s_0$
equals $\nu_{q}$ for every $q$
(and the probability of $\conf{q}{N}\btran{}s_1$ changes
appropriately to make the sum $1$).
Observe that since $\nu_q \leq  \vt{q}{N}$, for every $\conf{q}{i}$ where $i\leq N$:
\[
\sup_\sigma
\inf_\pi
\Pr{\sigma,\pi}{\conf{q}{i}}{\SSGReach{}\text{ in $\G'$}}
\leq
\sup_\sigma
\inf_\pi
\Pr{\sigma,\pi}{\conf{q}{i}}{\SSGReach{}\text{ in $\G$}}
=
\vt{q}{i}
.
\]
On the other hand, by our assumption on the values $\nu_q$ and $N$,
\[
\sup_\sigma
\inf_\pi
\Pr{\sigma,\pi}{\conf{q}{i}}{\SSGReach{}\text{ in $\G$}}
-\eps
\leq
\sup_\sigma
\inf_\pi
\Pr{\sigma,\pi}{\conf{q}{i}}{\SSGReach{}\text{ in $\G'$}}
.
\]
Thus
\(
\vt{q}{i}-
\sup_\sigma
\inf_\pi
\Pr{\sigma,\pi}{\conf{q}{i}}{\SSGReach{}\text{ in $\G'$}}
\leq \eps
,
\)
and we may output
\(
\nu_{\conf{q}{i}}
\coloneqq
\sup_\sigma
\inf_\pi
\Pr{\sigma,\pi}{\conf{q}{i}}{\SSGReach{}\text{ in $\G'$}}.
\)
By standard results, see, e.g.,~\cite{C92},
such reachability values have a binary encoding polynomial
in $\size{\G'}$, and after a memoryless optimal strategy (having size polynomial in $\size{\G'}$) is guessed, 
the values can be computed
in time polynomial in $\size{\G'}$.
If $\A$ is an OC-MDP, then $\G'$ is an MDP, and for MDPs the
reachability values, and optimal strategies, can be computed in deterministic polynomial
time.
Let us suppose we have computed
the optimal strategies $\sigma_R$, $\pi_R$ for reachability in $\G'$. The resulting strategy
$\bar\sigma$ for the given OC-SSG $\A$ is defined as follows:
In configurations with counter value between $0$ and $N$
it plays according to the optimal reachability strategy of Max
in $\G'$. Once a configuration with a counter value above $N$ is
visited it switches to playing as $\sigma^*$ forever, where $\sigma^*$ is
the optimal stategy we assume we are given for $\CN$.
Now for all configurations $\conf{q}{i}$, $0\leq i< N$, and strategies $\pi$ for Min,
the number $\nu_{\conf{q}{i}}$ gives a lower bound on the probability
that under $(\bar\sigma,\pi)$ a run either terminates without exceeding counter value $N$,
or hits some $\conf{r}{N}$ and then satisfies $\CN$.
This probability itself is a lower bound for the probability
that a run either terminates without exceeding counter value $N$,
or hits some $\conf{r}{N}$ and then terminates, which is in other words
the probability of termination.
This means that $\bar\sigma$ is $\eps$-optimal, because
$\vt{q}{i}-\nu_{\conf{q}{i}}\leq\eps$.

Analogously we define the strategy $\bar\pi$.
Consider again
some $\conf{q}{i}$, $0\leq i< N$, and $\sigma$ for Max.
The number $\nu_{\conf{q}{i}}$ gives now an upper bound on the probability
that under $(\sigma,\bar\pi)$ a run either terminates without exceeding counter value $N$,
or hits some $\conf{r}{N}$ and then satisfies $\CN$.
From the properties of $N$, this probability is by at most
$\eps$ lower than the probability of termination.
Because $\nu_{\conf{q}{i}}\leq\vt{q}{i}$ we obtain that also $\bar\pi$ is
$\eps$-optimal.
\myqed\end{proof}

\section{Conclusions}
\label{sec:concl}

We have shown that one can $\epsilon$-approximate the termination value
for OC-MDP (and for OC-SSG) termination games, and compute
$\epsilon$-optimal strategies for them,  in exponential time
(and in nondeterministic exponential time, respectively).

Our results leave open several intriguing problems.
An obvious remaining open problem is to obtain better complexity bounds. 
In particular, we know of no non-trivial lower bounds for OC-MDP approximation problems,
and it remains possible that approximation of the value for OC-MDPs can be computed in polynomial time. 
Our results also leave open the decidability of the quantitative
termination \emph{decision} problem 
for OC-MDPs and OC-SSGs, which asks: 
``is the termination value $\geq p$?'' for a given
rational probability $p$.    
Furthermore, our results leave open the computability of approximating the
value of \emph{selective termination} objectives for OC-MDPs, where the goal
is to terminate (reach counter value $0$) in a specific subset
of the control states. Qualitative versions of selective termination problems
were studied in \cite{BBEKW10,BBE10}.

\bibliography{bibliography}

\ifApp{
\newpage
\appendix
\section{Appendix}

\subsection{Non-existence of Optimal Strategies for Termination}
\label{sec:non-opt-ex}
In the following example we show that even in the special case
of OC-MDPs there may not be any optimal strategies for maximizing
the termination values. More precisely, there is a maximizing OC-MDP, $\A$,
and (infinitely many) configurations $\conf{s}{i}$ such that
for all strategies $\sigma$:
\(
\Pr\sigma{\conf{s}{i}}{\Term}
<
\vt{q}{i}
.
\)

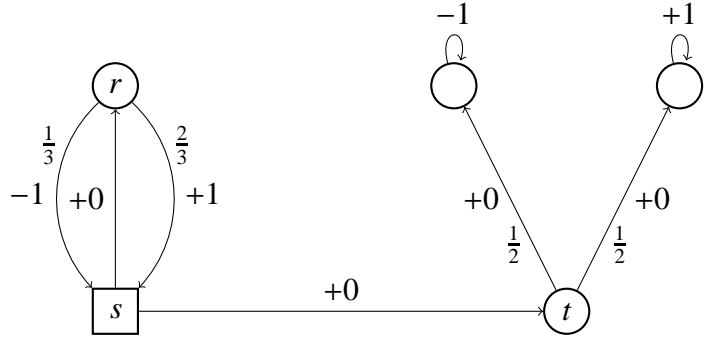
\begin{figure}
\begin{center}
\begin{tikzpicture}[x=3cm,y=3cm]
\node	(s)	at (0.5, 0)	[max]	{$s$};
\node	(r)	at (  0.5, 1)	[ran]	{$r$};
\node	(t)	at (2.5, 0)	[ran]	{$t$};
\node	(g)	at (  2, 1)	[ran]	{};
\node	(b)	at (  3, 1)	[ran]	{};

\path[->] (s) edge  node [left] {$+0$} (r)
 (r) edge [bend left=45] node [near start,right] {$\frac{2}{3}$} node [right] {$+1$} (s)
 (r) edge [bend right=45] node [near start,left] {$\frac{1}{3}$} node [left] {$-1$} (s)
 (s) edge node [above] {$+0$} (t)
 (t) edge node [near start,left] {$\frac{1}{2}$} node [left] {$+0$} (g)
 (t) edge node [near start,right] {$\frac{1}{2}$} node [right] {$+0$} (b)
 (g) edge [loop above] node [above] {$-1$} (g)
 (b) edge [loop above] node [above] {$+1$} (b);
\end{tikzpicture}
\end{center}
\caption{An OC-MDP where Player Max does not have optimal strategies for termination.
Signed numbers represent counter increments, unsigned are probabilities of transitions.}
\label{fig:ocmdp-non-max}
\end{figure}

\begin{example}
\label{ex:noopt}
Consider the maximizing OC-MDP, $\A$, given in Figure~\ref{fig:ocmdp-non-max}.
In the graph, round nodes represent stochastic states,
the only square node is a state of Player Max, $s$.
The arrows represent the rules, with signed numbers representing the
increments, and non-signed the probabilities.
For example the arrow from $s$ to $r$ represents the rule
$(s,0,r)$, whereas the right arrow from $r$ to $s$ represents
the rule $(r,+1,s)$, which has probability $P(r,1,s)=2/3$.

\begin{claim}
If the rule $(s,0,t)$ is removed then
$\vt{s}{i}=2^{-i}$.
\end{claim}
\begin{proof}
Observe that there is only one strategy when the rule
above is removed. We will omit writing its name.
Clearly
\(
\Pr{}{\conf{s}{0}}{\Term}=1=2^0.
\)
Further, the assignment
\(
x \coloneqq \Pr{}{\conf{s}{1}}{\Term}
\)
is the least non-negative solution
of the equation $x = \frac{1}{3} + \frac{2x^2}{3}$,
which is $\frac{1}{2}$.
Finally,
\(
\Pr{}{\conf{s}{i}}{\Term}
=
\frac{1}{3}
\cdot \Pr{}{\conf{s}{i-1}}{\Term}
+
\frac{2}{3}
\cdot \Pr{}{\conf{s}{i+1}}{\Term}.
\)
Given the initial conditions for $i=0,1$,
we obtain
\(
\Pr{}{\conf{s}{i}}{\Term}=2^{-i}
\)
as a unique solution of this recurrence.
\myqed\end{proof}

\begin{claim}
$\vt{s}{1}=\frac{3}{4}$.
\end{claim}
\begin{proof}
First we prove that
$\vt{s}{1}\geq\frac{3}{4}$.
For any $n$ consider the pure strategy, $\sigma_n$,
given for all histories ending in $\conf{s}{i}$ by
$\sigma_n(u)(\conf{s}{i}\tran{}\conf{r}{i})=1$ if $i<n$ and 
$\sigma_n(u)(\conf{s}{i}\tran{}\conf{t}{i})=1$ if $i\geq n$.
Set
\[
p_i \coloneqq \Pr{\sigma_i}{\conf{s}{1}}{\text{reach $\conf{s}{i}$}}.
\]
Observe that
$p_i$ stays the same number if we define it using any $\sigma_n$
with $n\geq i$, and that
\(
1-p_i = \Pr{\sigma_i}{\conf{s}{1}}{\text{terminate before reaching $\conf{s}{i}$}}.
\)
Moreover, $p_1=1$ and
\(
p_{i+1}
\coloneqq
\frac{2}{3} \cdot \left( p_i + (1-p_i)\cdot p_{i+1}\right).
\)
This uniquely determines that
\(
p_i = \frac{2^{i-1}}{2^i-1}.
\)
Note that $\lim_{i\to\infty} p_i=\frac{1}{2}$.
Finally, observe that
\[
\Pr{\sigma_n}{\conf{s}{1}}{\Term}
=
(1-p_n) + p_n\cdot\frac{1}{2}
\;.
\]
Thus, as $n\to\infty$ the probability
of termination under $\sigma_n$ approaches $\frac{3}{4}$.

Now we prove that
$\vt{s}{1}\leq\frac{3}{4}$
by proving that
\(
\Pr{\sigma}{\conf{s}{1}}{\text{terminate}}
\leq
\frac{3}{4}
\)
for every $\sigma$.
Consider the following probabilities:
\begin{align*}
p_a &\coloneqq \Pr\sigma{\conf{s}{1}}{\text{terminate without visiting $t$}},\\
p_b &\coloneqq \Pr\sigma{\conf{s}{1}}{\text{terminate after visiting $t$}},\\
p_c &\coloneqq \Pr\sigma{\conf{s}{1}}{\text{visit $t$}}.
\end{align*}
Clearly $p_b = \frac{p_c}{2}$.
Due to the first Claim, applied to $i=1$, we also have that $p_a \leq \frac{1}{2}$.
Finally, $p_a+p_c\leq 1$ since the events are disjoint.
We conclude that
\[
\Pr{\sigma}{\conf{s}{1}}{\Term}
=
p_a+p_b
\leq
p_a + \frac{1}{2}\cdot(1-p_a)
=
\frac{1}{2}\cdot p_a
+
\frac{1}{2}
\leq
\frac{3}{4}
.
\]
\myqed
\end{proof}

\begin{claim}
For all $i\geq0$,
\(
\vt{s}{i} = \frac{2^i+1}{2^{i+1}}.
\)
\end{claim}
\begin{proof}
The case $i=0$ is trivial, and $i=1$ is by the previous Claim.
Observe that $\vt{s}{i}\geq\frac{1}{2}$
for all $i$, because there is always the transition to
$\conf{t}{i}$ from where the system terminates with probability $\frac{1}{2}$.
Consequently, $\vt{r}{i}\geq\frac{1}{2}$ for all $i$ as well.

\begin{sloppypar}
Thus, for a fixed $i$, either
$\vt{s}{i}=\frac{1}{2}$.
or
$\vt{s}{i}>\frac{1}{2}$
In the first case, taking the transition $\conf{s}{i}\tran{}\conf{r}{i}$
is still value-optimal, i.e., $\vt{s}{i}=\frac{1}{2}\leq\vt{r}{i}$.
In the second case the transition
$\conf{s}{i}\tran{}\conf{t}{i}$
is not value-optimal, and
thus the transition
$\conf{s}{i}\tran{}\conf{r}{i}$
has to be value optimal.
\end{sloppypar}

Thus we know that 
$\conf{s}{i}\tran{}\conf{r}{i}$
always preserves the termination value, and we may unfold two steps of
the Bellman-style equations satisfied by the value to obtain
\[
\vt{s}{i}
=
\frac{1}{3}
\cdot \vt{s}{i-1}
+
\frac{2}{3}
\cdot \vt{s}{i+1}.
\]
Given the initial conditions for $i=0,1$,
we obtain
\(
\vt{s}{i} = \frac{2^i+1}{2^{i+1}}
\)
as a unique solution of this recurrence.
\myqed
\end{proof}

Thus for all $n\geq 1$ we have
$\vt{s}{n} = 2^{-(n+1)}\cdot(2^n+1)$,
and also, obviously, 
$\vt{t}{n} = 1/2$.
As a conclusion,
$\vt{s}{n} > \vt{t}{n}$.
Thus no termination-optimal strategy may choose a transition generated by the rule $(s,0,t)$.
On the other hand, as shown in the first Claim, without the rule $(s,0,t)$
we would have
$\vt{s}{n} = 2^{-n}<2^{-(n+1)}\cdot(2^n+1)$.
Consequently, there are no termination-optimal strategies in $\conf{s}{n}$.

\end{example}

\subsection{Reduction to Rising OC-MDPs}
\label{sec:no-idling}
Recall from Definition~\ref{def:idling} that a pure \stratCL{} strategy, $\sigma$,
is called \emph{idling}
if there is a state $q\in Q$, such that
\(
\Pr\sigma{\conf{q}{0}}{\exists i>0:\CState{i}=q}=1
\)
and for all $i\geq 0$:
\(
\Pr\sigma{\conf{q}{0}}{\CState{i}=q \implies \Cnt{i}=0}=1
.
\)
Also recall that a maximizing OC-MDP $\A$ is \emph{rising}
if there is no idling strategy for $\A$ and, moreover,
$\valO\CN{q}<1$ for all states $q$ of $\A$.
Before we start proving Lemma~\ref{lem:reduction} let us prove an auxiliary result.

\begin{lemma}
\label{lem:quadr}
Let $w$ be a finite path of length $n$ such that for all $\omega\in\runs{w}$:
\begin{itemize}
\item $\Cnt{i}(\omega)>\Cnt{0}(\omega)$ for all $i<n$, and
\item if $0\leq t<t'\leq n$ and $\CState{t}(\omega)=\CState{t'}(\omega)$
then $\Cnt{t}(\omega)>\Cnt{t'}(\omega)$.
\end{itemize}
Then $n\leq |Q|^2$
and $\max_{0\leq i \leq n}\Cnt{i}(\omega)-\Cnt{0}(\omega)\leq |Q|$
for all $\omega\in\runs{w}$.

\end{lemma}

\begin{proof}
From the fact that the maximal positive counter change is $+1$
and the second property of $w$, we have that
$\Cnt{i}(\omega)-\Cnt{0}(\omega)<|Q|$ for all $i<n$.
Again by the second property, every control state is thus visited
at most $|Q|$ times before the counter drops below $\Cnt{0}$.
By the first property we now have
$n\leq |Q|^2$ .
\myqed\end{proof}

\begin{proof}[Proof of Lemma~\ref{lem:reduction}]
We first intuitively explain the idea for the construction of $\bar\A$:
Using some added information in the control states, the OC-MDP will offer 
the following possibilities as long as a \stratCL{} strategy is chosen:
either the chosen \stratCL{} strategy somehow makes sure that
after any state $s$ is reached with positive probability 
the counter will thereafter either be decreased by at least one 
in at most $|Q|^2$ steps with positive probability,
or else after $s$ is reached the play will be forced to enter the
``trap'' state with positive probability.
The ``trap'' state is an extra absorbing state
that keeps increasing the counter value forever thereafter. 

This, firstly, ensures that given a OC-MDP, $\A$,
the newly constructed OC-MDP, $\bar\A$ that is derived from
it has no idling strategies.
Secondly, 
the construction ensures the following: for every state $q$ of
the original OC-MDP, $\A$, there is a corresponding state $\bar{q}$
of the newly constructed OC-MDP, $\bar\A$, such that the optimal termination
probability starting at configuration $\conf{q}{i}$ in $\A$ is equal to the 
optimal terimation probability starting at configuration $\conf{\bar{q}}{i}$ in
$\bar\A$.

In more detail,
the set $\bar{Q}$ of control states of $\bar\A$ will consist
of one special state ``$\qd$'', and of
multiple copies of states $Q$ enhanced with two counters.
These enhanced states are
3- and 5-tuples of the form
$\stateIII{q}{n}{m}$, $\stateV{q}{n}{m}{k}{r}$,
where $q \in Q$, $(q,k,r)\in\rules$,
$0\leq m\leq |Q|^2+1$ is a counter measuring
the number of steps until it exceeds $|Q|^2+1$, and
$0\leq n\leq |Q|+1$ is a counter measuring
the difference of the current counter value minus the initial one, until it drops below $0$
or goes above $|Q|$.

The triples and 5-tuples alternate in the transitions of
$\bar\A$. First comes a triple, $\stateIII{q}{n}{m}$, indicating
the current configuration of the simulation of a play
in $\A$. Then the player has to commit
to an outgoing rule, $(q,k,r)$, used on the short path 
(which it claims exists) which
decreases the counter.
This results in entering $\stateV{q}{n}{m}{k}{r}$.
If $q\in\sMax{Q}$ then the play must
move in the next step to $\stateIII{r}{n+k}{m+1}$.
If $q\in\sP{Q}$ then all outgoing rules for $q$
are used in the next step of the simulation, but the counters $m$ and $n$
are reset to $0$ for all steps following rules other than $(q,k,r)$.
Thus the next possible triples to visit are
$\stateIII{r}{n+k}{m+1}$,
corresponding to rule $(q,k,r)$, and 
states $\stateIII{r'}{0}{0}$  corresponding to rules 
$(q,k',r')$, where $(k',r')\neq (k,r)$.
The  state in $\bar\A$ corresponding to a $q$ in $\A$  is 
$\stateIII{q}{0}{0}$.  Starting at state $\stateIII{q}{0}{0}$
the states along a run in $\bar\A$ keep track
of the number of steps and the change in counter value, 
and if the number of steps ``overflows'' before the counter
decreases to $-1$, this
indicates that the
path selected by the player is not a short decreasing path,
and the simulation is aborted by transiting to the $\qd$ state, 
which results in an incrementing
self-loop.
Otherwise,  if within a short number of $m$ of steps we reach
a state $\stateV{q'}{0}{m}{q''}{k}$, where $m \leq |Q|^2$,
and the next transition decreases the counter (i.e., $k=-1$)
then the two ``internal counters'' are reset to $0$ and we start
all over again.

We now give a formal definition of $\bar\A$, which is an adaptation of
a similar contruction 
given in ~\cite{BBEKW10}, where it appeared as $\mathcal{D}'$.
The set of control states of $\bar\A$ is
\begin{multline*}
\bar{Q}
=
\{\qd\}
\cup
\{\stateIII{q}{n}{m}
\mid
q \in Q,
0\leq m\leq |Q|^2+1,
0\leq n\leq |Q|+1\}
\\
\{\stateV{q}{n}{m}{k}{r}
\mid
(q,k,r)\in\rules,
0\leq m\leq |Q|^2+1,
0\leq n\leq |Q|+1
\}
.
\end{multline*}
The stochastic states are
\(
\sP{\bar{Q}}
\coloneqq
\{\qd\}
\cup
\{
\stateV{q}{n}{m}{k}{r} \in \bar{Q}
\}
.
\)
The rules, $\bar\rules$, is the smallest set containing
\begin{align*}
&
\{
(\stateV{q}{n}{|Q|^2+1}{k}{r},1,\qd)
\mid (q,k,r)\in \rules, 0\leq n\leq |Q|+1
\}
\\
&
\cup
\{
(\stateV{q}{|Q|+1}{m}{k}{r},1,\qd)
\mid (q,k,r)\in \rules, 0\leq m\leq |Q|^2
\}
\\
&
\cup
\{
(\stateIII{q}{n}{m},0,\stateV{q}{n}{m}{k}{r}) \mid (q,k,r)\in\rules, 0\leq n\leq|Q|+1, 0\leq m\leq|Q|^2+1
\}
\\
&
\cup
\{
(\stateV{q}{n}{m}{k}{r},k,\stateIII{r}{n+k}{m+1}) \mid (q,k,r)\in\rules, 0\leq n\leq |Q|, n+k\geq0, 0\leq m\leq|Q|^2
\}
\\
&
\cup
\{
(\stateV{q}{n}{m}{k}{r},k,\stateIII{r}{0}{0}) \mid (q,k,r)\in\rules, n+k=-1, 0<m\leq|Q|^2
\}
\\
&
\cup
\{
(\stateV{q}{n}{m}{k}{r},k',\stateIII{r'}{0}{0}) \mid (q,k',r')\in\rules, q\in \sP{Q}, r'\neq r, 0\leq n \leq |Q|, 0\leq m \leq |Q|^2
\}
\\
&
\cup
\{
(\qd,1,\qd)
\}
,
\end{align*}
and also containing the rule $(\bar{q},1,\bar{q})$ for each state
not having an outgoing rule in the set above.
Finally, $\bar{P}$ is derived from $P$ as follows:
for all $\bar{q}\in\sP{\bar{Q}}$ which only have one outgoing rule
the probability of such rule is $1$.
Otherwise we know $\bar{q}=\stateV{q}{n}{m}{k}{r}$, $q\in\sP{Q}$ and can set
\(
\bar{P}((\stateV{q}{n}{m}{k}{r},k',\stateIII{r'}{n'}{m'}))=P((q,k',r'))
\)
for each
\(
(\stateV{q}{n}{m}{k}{r},k',\stateIII{r'}{n'}{m'})\in\bar\rules
.
\)

Clearly,
$\size{\bar\A}\in\mathcal{O}(\size{\A}^4).$
For $f$ we choose the function $f(q)=\stateIII{q}{0}{0}$.
The remaining three properties of $\bar\A$ are delivered
by Lemma~\ref{lem:no-idling},
Lemma~\ref{lem:val-no-change},
and Lemma~\ref{lem:pres-non-idling}.
\end{proof}

\begin{lemma}
\label{lem:no-idling}
There are no idling strategies in $\bar{A}$.
\end{lemma}
\begin{proof}
By contradiction, assume there is a pure \stratCL{} idling
strategy, $\sigma$.
From the definition of idling, there is a control state $\bar{q}\in\bar{Q}$ which
is almost surely revisited under $\sigma$, and upon every revisit, the counter
has the same value as at the beginning.
For every state $\bar{r}$ visited from $\bar{q}$, i.e., such that
$\Pr\sigma{\conf{\bar{q}}{0}}{\exists i\geq 0:\CState{i}=\bar{r}}>0$,
we define the set of possible counter values seen at a visit from $\conf{\bar{q}}{0}$ to $\bar{r}$
as
\(
C_{\bar{r}}
\coloneqq
\{
c\in\Zset
\mid
\exists i\geq 0:
\Pr\sigma{\conf{\bar{q}}{0}}{\CState{i}=\bar{r} \land \Cnt{i}=c}>0
\}
.
\)
First we observe that $|C_{\bar{r}}|=1$ for all such $\bar{r}$.
Indeed, it has obviously at least one element. On the other hand,
if $c,d\in C_{\bar{r}}$, $c\neq d$, then
\(
\Pr\sigma{\conf{\bar{q}}{0}}{\exists i>0: \CState{i}=\bar{q} \land \Cnt{i}=c-d}>0
\)
which contradicts our choice of $\bar{q}$ because $c-d\neq 0$.
From now on we denote by $c_{\bar{r}}$ the only number such that
$C_{\bar{r}}=\{c_{\bar{r}}\}$.

Now we choose $\bar{r}$ so that $c_{\bar{r}}$ is minimal.
Observe that
$\Pr\sigma{\conf{\bar{r}}{0}}{\Cnt{i}\geq 0}=1$ for all $i\geq 0$,
otherwise there is a state, $\bar{t}$, reachable under $\sigma$ from $\bar{r}$ such that
$c_{\bar{t}}<c_{\bar{r}}$.
But this means that a run from $\bar{r}$
under $\sigma$ visits $\qd$ almost surely.
It is easy to see that this implies that a run from $\bar{q}$
visits $\qd$ with a positive probability.\footnote{Actually this probability
is again $1$, but we only need to know that it is positive.}
This contradicts $\sigma$ being idling and $\bar{q}$ being the witnessing state for idling.
\myqed\end{proof}

Before we prove the second important property of $\bar\A$ we promised,
we note that although technically it is not true that $Q\subseteq \bar{Q}$,
we may insert $Q$ into $\bar{Q}$ by mapping $q$ to $\stateIII{q}{0}{0}$.
\begin{lemma}
\label{lem:val-no-change}
\(
\vt{q}{i}
=
\vt{\stateIII{q}{0}{0}}{i}
\)
for all $q\in Q$ and $i\geq 0$.
\end{lemma}
\begin{proof}
The inequality
\(
\vt{q}{i}
\geq
\vt{\stateIII{q}{0}{0}}{i}
\)
is easy, because a strategy in $\A$ can simulate
a strategy in $\bar\A$ (by ``projecting'' it onto states of $\A$), except for the case
when the run in $\bar\A$ reaches $\qd$. But after
reaching $\qd$ no run terminates, so the simulation in $\A$
may continue arbitrarily without producing a lower probability of termination.

To prove
\(
\vt{q}{i}
\leq
\vt{\stateIII{q}{0}{0}}{i}
,
\)
we need to show that there are $\eps$-optimal strategies for $\A$,
for arbitrarily small $\epsilon > 0$, 
which can be simulated in $\bar\A$ while keeping the termination
probability $\eps$-close to the original optimal termination value in $\A$.
In the simulation we will use a natural correspondence of paths in $\bar\A$
to paths in $\A$, given by dropping the odd steps and all additional information.
As an example, the path
\(
\conf{\stateIII{q}{0}{0}}{0}
\tran{}
\conf{\stateV{q}{0}{0}{+1}{r}}{0}
\tran{}
\conf{\stateIII{r}{1}{1}}{1}
\)
corresponds to
\(
\conf{q}{0}
\tran{}
\conf{r}{1}
.
\)

In the proof, we give for every $\eps>0$ a pure strategy $\sigma_\eps$,
and a measurable set of runs, $T_\eps \subseteq \Term$  in $\A$, such that
for all $q\in Q$ and $i\geq 0$:
\begin{itemize}
\item
$\Pr{\sigma_\eps}{\conf{q}{i}}{T_\eps}\geq\vt{q}{i}-\eps$, and
\item
for all finite paths $u$, $\len{u}=n$,
such that
$\Pr{\sigma_\eps}{\conf{q}{i}}{\runs{u}\cap T_\eps}>0$,
there is some $k$, $n<k\leq n+|Q|^2+1$ for which
\end{itemize}
\begin{equation}
\label{eq:shortP}
\Pr{\sigma_\eps}{\conf{q}{i}}{
\Cnt{k}<\Cnt{n}
\land
\forall j, n<j<k:
0\leq\Cnt{j}-\Cnt{n}\leq|Q|
\mid \runs{u}
}>0
.
\end{equation}
Once we have proved the above, we can simulate the strategy $\sigma_\eps$
in $\bar\A$.

Let us define the simulation in detail.
Let $\bar{u}$ be a path from configuration 
$\conf{\stateIII{q}{0}{0}}{i}$ in $\bar\A$, ending in
some configuration $\conf{\stateIII{r}{0}{0}}{j}$, and $u$ the corresponding path in $\A$,
ending in $\conf{r}{j}$.
Let $n = \len{u}$. If 
$\Pr{\sigma_\eps}{\conf{q}{i}}{\runs{u}\cap T_\eps}=0$,
then the rest of the simulating strategy after intial path $\bar{u}$ 
can be defined arbitrarily.
For later reference we call $\bar{u}$ and all its extensions
\emph{dead} in this case.
Otherwise
let $w$ be some extension of $u$ witnessing (\ref{eq:shortP}), i.e.,
$w$ of length $k\leq n+|Q|^2+1$ with a prefix $u$,
such that
\(
\Pr{\sigma_\eps}{\conf{q}{i}}{\runs{w}}>0
\)
and
\(
\Pr{\sigma_\eps}{\conf{q}{i}}{
\Cnt{k}<\Cnt{n}
\land
\forall j, n<j<k:
0\leq\Cnt{j}-\Cnt{n}\leq|Q|
\mid \runs{w}
}=1
.
\)

We fix a unique choice of such a $w$ (which depends on $u$), and we 
define the simulating strategy in $\bar\A$ for all
histories $\bar{v}$ such that the path $v$ in $\A$ which corresponds
to $\bar{v}$ is an extension of $u$ and a proper prefix of $w$.
The definition is by induction on the length of $\bar{v}$.
Such a history $\bar{v}$ ends in some $\stateIII{s}{h}{m}$, where $0\leq m<|Q|^2+1$
and $h\geq 0$.
Let $\conf{s}{c}$ be the last configuration on $v$ and
$\conf{s}{c}\tran{}\conf{s'}{c'}$ the next step in $w$ after completing $v$.
Then the rule chosen with probability $1$ by the simulating strategy 
in $\bar{A}$ for the history $\bar{v}$
is $(\stateIII{s}{h}{m},0,\stateV{s}{h}{m}{c'-c}{s'})$.
Observe that due to our choice of $w$, 
the control state visited in the simulation after completing
$\bar{v}$ and visiting $\stateV{s}{h}{m}{c'-c}{s'}$
is either (a) 
$\stateIII{s'}{h+c'-c}{m+1}$, where $m+1<|Q|^2+1$ and $h+c'-c\geq 0$,
or else (b) a state $\stateIII{s''}{0}{0}$, for some state $s'' \in Q$.
In the former case (a) we continue with a new $\bar{v}$ as above.
In the latter case (b), we are again
back in a state of the form $\stateIII{r}{0}{0}$,
and thus we need to find a new extension $w'$ 
(unless now we are in a dead history)
and start the process all over again.
Because every history in the simulation is either dead, or ends in some
$\conf{\stateIII{r}{0}{0}}{j}$, or is some short extension $\bar{v}$ of such a
history which is not dead and ends in some state $\stateIII{r}{0}{0}$, 
as above, we have now defined the simulating strategy for every history 
in $\bar\A$.

Moreover, consider a path $\bar{u}$ in $\bar\A$, which is not dead.
Because we could not possibly hit $\qd$ in $\bar\A$ before reaching
a dead history, and because $\sigma_\eps$ is pure,
the probability of $\runs{\bar{u}}$ in the simulation
is $\Pr{\sigma_\eps}{\conf{q}{i}}{\runs{u}}$, where $u$ is the path corresponding
to $\bar{u}$ in $\A$.
As a consequence, once we prove the existence of $\sigma_\eps$ and validity of
its properties, we have proven that the termination value in the simulation is
at least $\vt{q}{i}-\eps$.
Because $\eps>0$ can be chosen arbitrarily, this proves
\(
\vt{q}{i}
\leq
\vt{\stateIII{q}{0}{0}}{i}
.
\)
In the rest of the proof we show how to construct such a 
strategy $\sigma_\eps$ in $\A$
for every $\epsilon > 0$ .

Let $t\geq 0$.
By $\TermB{t}$ we denote the event that $\Cnt{t'}=0$ for some $t'\leq t$.
By standard facts (see, e.g.,~\cite[Theorem~4.3.3]{Puterman94}),
for all $t$ there is a pure strategy $\tau_t$ optimal for $\TermB{t}$,
i.e., such that for all $q\in Q$, $i\geq 0$:
$\Pr{\tau_t}{\conf{q}{i}}{\TermB{t}} = \vtb{q}{i}{t}$.
Also, easily
\(
\lim_{t\to\infty} \vtb{q}{i}{t}
=
\vt{q}{i}
,
\)
thus for all $\eps>0$ there is $t_\eps$ such that
\(
\vtb{q}{i}{t_\eps}
\geq
\vt{q}{i} - \eps
.
\)
We set $T_\eps \coloneqq \TermB{t_\eps}$.

Let us fix an $\eps>0$, and consider the corresponding $t_\eps$.
We now define $\sigma_\eps$.
Let $u$ be a path in $\A$, of length $n<t_\eps$, ending
in a configuration $\conf{r}{j}$.
Pick the least $t\leq t_\eps-n$ such that
\(
\vtb{q}{i}{t}
=
\vtb{q}{i}{t_\eps-n}
.
\)
Then $\sigma_\eps(u)=\tau_t(\conf{r}{j})$.
For $u$ where $\len{u}\geq t_\eps$ we define $\sigma_\eps(u)$ arbitrarily.
Due to the Bellman-equation characterization of optimality for finite-horizon objectives,
given, e.g., in~\cite[Section~4.3]{Puterman94}, we obtain that
for all configurations $\conf{q}{i}$:
\(
\Pr{\sigma_\eps}{\conf{q}{i}}{\Term}
\geq
\Pr{\sigma_\eps}{\conf{q}{i}}{\TermB{t_\eps}}
=
\vtb{q}{i}{t_\eps}
\geq
\vt{q}{i}-\eps
,
\)
as required.

For $k\geq0$, let $E_k$ be the event that there are times $t,t'$,
$k\leq t<t'$,
such that $\CState{t}=\CState{t'}$
and $0<\Cnt{t}\leq\Cnt{t'}$.
\begin{claim}
\label{cl:Ek}
\begin{sloppypar}
Let $k\geq 0$, and $\conf{q}{i}$ be a configuration.
Further, let $u$ be an arbitrary path such that $\len{u}=k$,
\(
\Pr{\sigma_\eps}{\conf{q}{i}}{\runs{u}}>0
\)
and
\(
\Pr{\sigma_\eps}{\conf{q}{i}}{\TermB{t_\eps} \mid \runs{u}}>0
.
\)
Then
$\Pr{\sigma_\eps}{\conf{q}{i}}{E_k \mid \runs{u}\cap\TermB{t_\eps}}<1$.
\end{sloppypar}
\end{claim}
\begin{proof}
\begin{sloppypar}
By contradiction.
For $k\geq t_\eps$ the statement is obvious.
Fix some $k,\ 0\leq k< t_\eps$ and $\conf{q}{i}$.
Assume that
$\Pr{\sigma_\eps}{\conf{q}{i}}{E_k \mid \runs{u}\cap\TermB{t_\eps}}=1$.
Let $w$ be an arbitrary extension of $u$ such that
$\len{w}=t_\eps$,
$\runs{w}\subseteq\TermB{t_\eps}$, and
$\Pr{\sigma_\eps}{\conf{q}{i}}{\runs{w}}>0$.
Then clearly $\runs{w}\subseteq E_k$.
This means that there are times $t,t'$,
$k\leq t < t'\leq t_\eps$
such that
$\Pr{\sigma_\eps}{\conf{q}{i}}{\CState{t}=\CState{t'} \mid \runs{w}}=1$, and
$\Pr{\sigma_\eps}{\conf{q}{i}}{0<\Cnt{t}\leq \Cnt{t'} \mid \runs{w}}=1$.
Consider the prefixes $\bar{w}$, $\bar{w}'$ of $w$ of lengths
$t$ and $t'$, respectively. There is some state $r\in Q$, and counter values $0<j\leq j'$
such that $\bar{w}$ ends in $\conf{r}{j}$, and $\bar{w}'$ ends in $\conf{r}{j'}$.
By the construction of $\sigma_\eps$, there are
$h\leq t_\eps-t$ and $h'\leq t_\eps-t'$ such that
$h'<h$, $\vtb{r}{j}{h}>\vtb{r}{j}{h-1}\geq\vtb{r}{j'}{h'}$, and
\(
\Pr{\sigma_\eps}{\conf{q}{i}}{\TermB{t_\eps} \mid \runs{\bar{w}}}
=
\vtb{r}{j}{h}
,
\)
\(
\Pr{\sigma_\eps}{\conf{q}{i}}{\TermB{t_\eps} \mid \runs{\bar{w}'}}
=
\vtb{r}{j'}{h'}
.
\)
In other words, on every extension of $u$ which eventually satisfies
$\TermB{t_\eps}$ there is a moment where the probability of $\TermB{t_\eps}$,
conditionally on the current history, sharply decreases.
This is in contradiction with the fact that $\sigma_\eps$ is optimal
wrt.\ $\TermB{t_\eps}$, and thus satisfies the Bellman optimality equations
(cf.~\cite[Section~4.3]{Puterman94}).
\end{sloppypar}
\end{proof}

Let us fix an arbitrary path $u$, $\len{u}=n$, such that
$\Pr{\sigma_\eps}{\conf{q}{i}}{\runs{u}\cap\TermB{t_\eps}}>0$.
We apply Claim~\ref{cl:Ek}, and obtain a witnessing extension, $w$, $\len{w}=m$, of
$u$ so that $\Pr{\sigma_\eps}{\conf{q}{i}}{\runs{w}}>0$,
$\Pr{\sigma_\eps}{\conf{q}{i}}{\Cnt{m}=0 \mid \runs{w}}=1$, and
$\Pr{\sigma_\eps}{\conf{q}{i}}{E_n \mid \runs{w}}=0$.
By
Lemma~\ref{lem:quadr} this implies that
there is some $k$, $n<k\leq n+|Q|^2+1$
such that (\ref{eq:shortP}) is satisfied.
Thus we proved all the required properties of $\sigma_\eps$ and
$T_\eps = \TermB{t_\eps}$, and the proof is done.
\myqed\end{proof}

As a consequence, we obtain the last promised property of $\bar\A$.
\begin{lemma}
\label{lem:pres-non-idling}
If $\valO\CN{q}<1$ for all $q\in Q$ in $\A$, then $\bar\A$ is rising.
\end{lemma}

\begin{proof}
By Lemma~\ref{lem:no-idling} there are no idling strategies in $\bar\A$.
It remains to prove that
$\valO\CN{\bar{q}}<1$ for all $q\in \bar{Q}$ in $\bar\A$.
First we prove it for $\bar{q}$ of the form $\stateIII{q}{0}{0}$.
If $\valO\CN{q}<1$ then there is $i\geq 0$ such that
$\vt{q}{i}<1$, by, e.g., Lemma 14 of \cite{BBE10}.
By Lemma~\ref{lem:val-no-change} we thus have
$\vt{\bar{q}}{i}<1$, and, thus again by Lemma 14 of~\cite{BBE10},
$\valO\CN{\bar{q}}<1$.
If $\bar{q}=\qd$ we are done immediately, as
$\valO\CN{\qd}=0$.
Finally, in all remaining cases of $\bar{q}$ the play will almost surely
reach some states from $\{\qd\}\cup\{\stateIII{q}{0}{0} \mid q\in Q\}$.
Because $\CN$ is prefix independent,
$\valO\CN{\bar{q}}<1$ also in this case, and the proof is finished.
\end{proof}

\subsection{Bounds on $N$}
\label{sec:app-bounds}

Here we derive an exponential upper bound on the value $N$, introduced
in Section~\ref{sec:res}.
Recall that, given a OC-SSG, $\A=\ocssgA$, and an $\eps>0$, we want $N$
to satisfy:
\[
\vt{q}{i}
-
\text{Val}(\CN,q)
\leq
\eps
\quad
\text{for all $q\in Q$ and $i\geq N$.}
\]
By results of Section~\ref{sec:OC-SSG}, it suffices to consider only
the case when $\A$ is a maximizing OC-MDP.
From Section~\ref{sec-OC-MDP} we know that
\(
N
\coloneqq
\max\{
h,
\lceil
\log_c(\eps\cdot(1-c))
\rceil
\}
,
\)
where
\(
c
=
\exp
\left(
\frac%
{-\bar{x}^2}%
{2\cdot(\vdiff+\bar{x}+1)}
\right)
<1
\)
and
\(
h
=
\lceil \vdiff \rceil
,
\)
and $\bar{x}$ and $\vdiff$ are solutions to a linear program
with coefficients polynomial in $\size\A$.
Thus there is a positive polynomial, $p$, such that
\(
c
\leq
\expo{-\expo{-p(\size\A)})}
\)
and
\(
h
\leq
\expo{p(\size\A)}
.
\)
If $N\leq h$ we have clearly that it is exponentially bounded
in $\size\A$.
Otherwise
\[
N
\leq
f_\eps(c)
\coloneqq
\frac{\ln(\eps) + \ln(1-c)}{\ln(c)}.
\]
Observe that $f_\eps(c)$ is growing with $c\to1^{-}$ and fixed $\eps$,
because
$\frac{c^{f_\eps(c)}}{1-c}=\eps$
and
$\frac{c^{i}}{1-c}$ grows with $c\to1^{-}$ and fixed $i$.
Thus
\begin{multline}
\label{eq:Nb}
N
\leq
f_\eps(c)
\leq
f_\eps(\expo{-\expo{-p(\size\A)}})
\\
=
\frac{\ln(\eps) + \ln(1-\expo{-\expo{-p(\size\A)}})}{-\expo{-p(\size\A)}}
=
\expo{p(\size\A)}
\cdot
\ln(1/\eps) 
-
\ln(1-\expo{-\expo{-p(\size\A)}})
\cdot
\expo{p(\size\A)}
.
\end{multline}

Before we prove that this is indeed an exponential bound on $N$,
let us prove two auxiliary claims.

\begin{claim}
For all $n\geq 0$ the following inequality holds:
\begin{equation}
\label{eq:exp-gap}
\expo{-1}-\expo{-1-\expo{-n}}
\leq
1-\expo{-\expo{-(n+1)}}
.
\end{equation}
\end{claim}
\begin{proof}
We set
\(
d(n)
\coloneqq
\expo{-\expo{-(n+1)}}
-
\expo{-1-\expo{-n}}
.
\)
The inequality (\ref{eq:exp-gap}) is equivalent to $d(n)\leq 1-\expo{-1}$.
Because $\lim_{n\to\infty}d(n)=1-\expo{-1}$, it suffices to
prove that $d(n)$ is increasing:
Observe that
\begin{equation}
\label{eq:dif}
d(n+1)-d(n)
=
(
\expo{-\expo{-(n+2)}}
-
\expo{-\expo{-(n+1)}}
)
-\expo{-1}\cdot
(
\expo{-\expo{-(n+1)}}
-
\expo{-\expo{-n}}
)
.
\end{equation}
Also, because the exponential function $\expo{x}$ is increasing and has increasing derivation $\expo{x}\geq 0$,
we know that
\[
\frac{
\expo{a}-\expo{b}
}{
\expo{b}-\expo{c}
}
\geq
\frac{
a-b
}{
b-c
}
\quad
\text{for all $a>b>c$.}
\]
In particular, setting
$a=-\expo{-(n+2)}$,
$b=-\expo{-(n+1)}$, and
$c=-\expo{-n}$
yields
\[
\frac{
\expo{-\expo{-(n+2)}}-\expo{-\expo{-(n+1)}}
}{
\expo{-\expo{-(n+1)}}-\expo{-\expo{-n}}
}
\geq
\expo{-1}
.
\]
By (\ref{eq:dif}), this implies $d(n+1)\geq d(n)$ as required.
\myqed\end{proof}

\begin{claim}
For all $n\geq 0$ the following inequality holds:
\begin{equation}
\label{eq:lin}
n+1\geq -\ln(1-\expo{-\expo{-n}}).
\end{equation}
\end{claim}
\begin{proof}
By induction.
A direct computation for $n=0$ shows
\(
-\ln(1-\expo{-\expo{-0}})
=
-\ln(1-\expo{-1})
\leq
0.46
<
1
.
\)
Consider now $n=k+1$ for some $k\geq 0$.
Using (\ref{eq:exp-gap}) and the inductive hypothesis, we obtain
\[
(k+1)+1
\geq
-\ln(1-\expo{-\expo{-k}}) + 1
=
-\ln(\expo{-1}-\expo{-1-\expo{-k}})
\geq
-\ln(1-\expo{-\expo{-(k+1)}})
.
\]
\myqed\end{proof}

Finally, using (\ref{eq:lin}) in (\ref{eq:Nb}) we get
\(
N
\leq
\expo{p(\size\A)}
\cdot
\ln(1/\eps) 
+
(1+p(\size\A))
\cdot
\expo{p(\size\A)}
.
\)

}{}

\end{document}